\documentclass[11pt]{article}
\usepackage{amsthm,amssymb,amsmath}
\usepackage{graphicx,algorithmic,algorithm}%
\usepackage{enumerate}
\usepackage{tikz,tikz-cd}
\tikzcdset{scale cd/.style={every label/.append style={scale=#1},
    cells={nodes={scale=#1}}}}%
\usetikzlibrary{shapes}
\usepackage[letterpaper, margin=1in]{geometry}
\usepackage{hyperref}
\usepackage{todonotes}
\usepackage[normalem]{ulem}
\hypersetup{
	pdftitle=   {},
	pdfauthor=  {Yeonsu Chang, O-joung Kwon and Myounghwan Lee}
}
\usepackage{authblk}
\usepackage{mathabx}

\newtheorem{theorem}{Theorem}[section]
\newtheorem{lemma}[theorem]{Lemma}

\newtheorem{corollary}[theorem]{Corollary}

\newtheorem{claim}[theorem]{Claim}

\newtheorem*{thm:main1}{Theorem~\ref{thm:main1}}
\newtheorem*{thm:main2}{Theorem~\ref{thm:main2}}
\newenvironment{clproof}{\begin{list}{}{%
			\setlength{\leftmargin}{3mm}%
		} \item {\it Proof.} }{\hfill$\lozenge$\end{list}}

\newcommand\extrafootertext[1]{%
    \bgroup
    \renewcommand\thefootnote{\fnsymbol{footnote}}%
    \renewcommand\thempfootnote{\fnsymbol{mpfootnote}}%
    \footnotetext[0]{#1}%
    \egroup
}

\newcommand\abs[1]{\lvert #1\rvert}

\newcommand\ecw{\operatorname{ecw}}

\newcommand\secw{\operatorname{sec}}

\newcommand\tw{\operatorname{tw}}
\newcommand\carvw{\operatorname{carvw}}
\newcommand\fen{\operatorname{fen}}
\newcommand\tcw{\operatorname{tcw}}
\newcommand\stcw{\operatorname{stcw}}
\newcommand\tpw{\operatorname{tpw}}
\newcommand\ecrw{\operatorname{ecrw}}

\newcommand\cross{\operatorname{cross}}
\newcommand\adh{\operatorname{adh}}

\newcommand\WR{\preccurlyeq}

\newcommand{\specialcell}[2][c]{%
  \begin{tabular}[#1]{@{}c@{}}#2\end{tabular}}
  
\begin{document}
	\title{A new width parameter of graphs based on edge cuts: $\alpha$-edge-crossing width}
	
	\author{Yeonsu Chang}
	\author[1,2]{O-joung Kwon}
	\author{Myounghwan Lee}

	\affil{Department of Mathematics, Hanyang University, Seoul, South Korea.}
	\affil[2]{Discrete Mathematics Group, Institute for Basic Science (IBS), Daejeon, South Korea}
	
	\date\today
	\maketitle

	\extrafootertext{E-mail addresses: \texttt{yeonsu@hanyang.ac.kr} (Y. Chang), \texttt{ojoungkwon@hanyang.ac.kr} (O. Kwon) and \texttt{sycuel@hanyang.ac.kr} (M. Lee) }
	
	\begin{abstract}
	We introduce graph width parameters, called $\alpha$-edge-crossing width and edge-crossing width. These are defined in terms of the number of edges crossing a bag of a tree-cut decomposition. They are motivated by edge-cut width, recently introduced by Brand et al. (WG 2022).
    We show that edge-crossing width is equivalent to the known parameter tree-partition-width.
    On the other hand, $\alpha$-edge-crossing width is a new parameter; tree-cut width and $\alpha$-edge-crossing width are incomparable, and they both lie between tree-partition-width and edge-cut width. 

We provide an algorithm that, for a given $n$-vertex graph $G$ and integers $k$ and $\alpha$, in time $2^{O((\alpha+k)\log (\alpha+k))}n^2$ either outputs a tree-cut decomposition certifying that the $\alpha$-edge-crossing width of $G$ is at most $2\alpha^2+5k$ or confirms that the $\alpha$-edge-crossing width of $G$ is more than~$k$.
     As applications, 
 	for every fixed~$\alpha$, we obtain FPT algorithms for the \textsc{List Coloring} and \textsc{Precoloring Extension} problems parameterized by $\alpha$-edge-crossing width. They were known to be W[1]-hard parameterized by tree-partition-width, and FPT parameterized by edge-cut width, and we close the complexity gap between these two parameters.

	\end{abstract}

	\section{Introduction}\label{sec:intro}

    Tree-width is one of the basic parameters in structural and algorithmic graph theory, which measures how well a graph accommodates a decomposition into a tree-like structure.
    It has an important role in the graph minor theory developed by Robertson and Seymour~\cite{RS1986, RS1991, RS2004}.
    For algorithmic aspects, there are various fundamental problems that are NP-hard on general graphs, but fixed parameter tractable (FPT) parameterized by tree-width, that is, that can be solved in time $f(k)n^{O(1)}$ on $n$-vertex graphs of tree-width $k$ for some computable function $f$. However, various problems are still W[1]-hard parameterized by tree-width. 
    For example, \textsc{List Coloring} is W[1]-hard parameterized by tree-width~\cite{FellowsF2011}.

    Recently, edge counterparts of tree-width have been considered. One of such parameters is the \emph{tree-cut width} of a graph introduced by Wollan~\cite{Wollan2015}. Similar to the relationship between tree-width and graph minors, Wollan established a relationship between tree-cut width and weak immersions, and discussed structural properties. 
    Since tree-cut width is a weaker parameter than tree-width, one could expect that some problems that are W[1]-hard parameterized by tree-width, are fixed parameter tractable parameterized by tree-cut width. 
    But still several problems, including \textsc{List Coloring}, remain W[1]-hard parameterized by tree-cut width~\cite{DidemSC2017, GanianO2021, GanianK2021, RobertKDR2022, GanianKS2022}.

\begin{figure}[t]
    \centering
    \[\begin{tikzcd}[column sep=scriptsize, scale cd=0.85, row sep=large]
    &  &  &  & \mathrm{fen}\arrow[d,"\text{\cite{BrandCHGK2022}}"]\\
    \mathrm{carvw}\arrow[rrd,"\text{\cite{Ganian2022}}"]&  &  &  & \mathrm{ecw}\arrow[lld,"\text{\cite{Ganian2022}}"']\\
    &  & \mathrm{stcw}\arrow[lld,"\text{\cite{Ganian2022}}"']\arrow[rrd,"\text{Lem. 3.10}"] &  & \\    \mathrm{tcw}\arrow[rrd,"\text{\cite{giannopoulou2017packing}}"]&  &  &  & \mathrm{ecrw}_\alpha\arrow[lld,"\text{Lem. 3.2}"]\\
    & & \specialcell{$\mathrm{tpw}\arrow[d,"\text{\cite{seese1985tree}}"]\sim \mathrm{ecrw}$  \\ (\small Lem. 3.11 \& Lem. 3.18) } & &  & \\
    & & \mathrm{tw}\arrow[d,"\text{\cite{dallard2024treeindep}}"] &  & \\
    & & \mathrm{tree}\text{-}\alpha & &
\end{tikzcd}\]
    \caption{The hierarchy of the mentioned width parameters. For two width parameters $A$ and $B$, $A\to B$ means that every graph class of bounded $A$ has bounded $B$, but there is a graph class of bounded $B$ and unbounded $A$. Also, $A\sim B$ means that two parameters $A$ and $B$ are asymptotically equivalent. $\fen$, $\carvw$, $\ecw$, $\tcw$, $\stcw$, $\ecrw_\alpha$, $\ecrw$, $\tpw$, $\tw$, and $\mathrm{tree}$-$\alpha$ denote feedback edge set number, carving-width, edge-cut width, tree-cut width, slim tree-cut width, $\alpha$-edge-crossing width, edge-crossing width, tree-partition-width, tree-width, and tree-independence number, respectively.}
    \label{fig:parameters}
\end{figure}

    This motivates Brand et al.~\cite{BrandCHGK2022} to consider a more restricted parameter called the \emph{edge-cut width} of a graph. For the edge-cut width of a graph $G$, the authors considered maximal spanning forests $F$ of $G$. For each vertex $v$ of $F$, the \emph{local feedback edge set} of $v$ is the number of edges $e\in E(G)\setminus E(F)$ where the unique cycle of the graph obtained from $F$ by adding $e$ contains $v$, and the \emph{edge-cut width} of $F$ is the maximum local feedback edge set plus one over all vertices of $G$. The edge-cut width of $G$ is the minimum edge-cut width among all maximal spanning forests of $G$. The edge-cut width with respect to a maximal spanning forest was also considered by Bodlaender~\cite{Bodlaender1998} to bound the tree-width of certain graphs, with a different name called \emph{vertex remember number}. Brand et al. showed that the tree-cut width of a graph is at most its edge-cut width. Furthermore, they showed that several problems including \textsc{List Coloring} are fixed parameter tractable parameterized by edge-cut width.

    A natural question is to find a width parameter $f$ such that graph classes of bounded $f$ strictly generalize graph classes of bounded edge-cut width, and also \textsc{List Coloring}  admits a fixed parameter tractable algorithm parameterized by $f$. 
    This motivates us to define a new parameter called \emph{$\alpha$-edge-crossing width}. By relaxing the condition, we also define a parameter called \emph{edge-crossing width}, but it turns out that this parameter is equivalent to tree-partition-width~\cite{tpw1996}.
    Recently, Ganian and Korchemna~\cite{Ganian2022} introduced slim tree-cut width which also generalizes edge-cut width.
    See Figure~\ref{fig:parameters} for the hierarchy of new parameters and known parameters.

    We define the \emph{$\alpha$-edge-crossing width} and \emph{edge-crossing width} of a graph.
    For a graph $G$, a pair $\mathcal{T}=(T, \mathcal{X})$ of a tree $T$ and a collection $\mathcal{X}=\{X_t\subseteq V(G):t\in V(T)\}$ of disjoint sets of vertices in $G$, called bags (allowing empty bags), with the property $\bigcup_{t\in V(T)}X_t=V(G)$ is called the \emph{tree-cut decomposition} of $G$.
    For a node $p\in V(T)$, let $T_1, T_2, \cdots, T_m$ be the connected components of $T-p$, and let $\cross_{\mathcal{T}}(p)$ be the number of edges incident with two distinct sets in $\{\bigcup_{t\in V(T_i)}X_t :1\le i\le m\}$. Every edge $ab$ of $G$, where $a$ and $b$ belong to distinct sets in $\{\bigcup_{t\in V(T_i)}X_t :1\le i\le m\}$, is said to \emph{cross $X_p$}.
    The \emph{crossing number} of $\mathcal{T}$ is $\max_{p\in V(T)}\cross_{\mathcal{T}}(p)$, and the \emph{thickness} of $\mathcal{T}$ is $\max_{p\in V(T)}\abs{X_p}$.
     For a positive integer $\alpha$, the \emph{$\alpha$-edge-crossing width} of  a graph $G$, denoted by $\ecrw_\alpha(G)$, is the minimum crossing number over all tree-cut decompositions of $G$ whose thicknesses are at most $\alpha$.
   The \emph{edge-crossing width} of $\mathcal{T}$ is the maximum of the crossing number and the thickness of $\mathcal{T}$.
    The \emph{edge-crossing width} of $G$, denoted by $\ecrw(G)$, is the minimum \emph{edge-crossing width} over all tree-cut decompositions of $G$.

    It is not difficult to see that the $1$-edge-crossing width of a graph is at most its edge-cut width minus one, as we can take the completion of its optimal maximal spanning forest for edge-cut width into a tree as a tree-cut decomposition with small crossing number.

    We provide an FPT approximation algorithm for $\alpha$-edge-crossing width. We adapt an idea for obtaining an FPT approximation algorithm for tree-cut width due to Kim et al~\cite{KimOP2018}.
    \begin{theorem}\label{thm:approxalpha}
    Given an $n$-vertex graph $G$ and two positive integers $\alpha$ and $k$, one can in time $2^{\mathcal{O}\left((\alpha+k)\log(\alpha+k)\right)}n^2$ either 
    \begin{itemize}
        \item output a tree-cut decomposition of $G$ with thickness at most $\alpha$ and crossing number at most $2\alpha^2+5k$, or
        \item correctly report that $\ecrw_\alpha(G)>k$.
    \end{itemize}
    \end{theorem}

   As applications of $\alpha$-edge-crossing width, we show that \textsc{List Coloring} and \textsc{Precoloring Extension} are FPT parameterized by $\alpha$-edge-crossing width. 
   They were known to be W[1]-hard parameterized by tree-cut width (and so by tree-partition-width)~\cite{GanianKS2022}, and FPT parameterized by edge-cut width~\cite{BrandCHGK2022} and by slim tree-cut width~\cite{Ganian2022}. We close the complexity gap between these parameters.

    \begin{theorem}\label{thm:listcoloring}
    For a fixed positive integer $\alpha$, the \textsc{List Coloring} and \textsc{Precoloring Extension} problems are FPT parameterized by $\alpha$-edge-crossing width.
    \end{theorem}

    A key idea of the algorithm is to use the representative set technique for coloring problems. This can be considered a point of independent interest.

    This paper is organized as follows. 
    In Section~\ref{sec:prelim}, we give basic definitions and notations.
    We present an FPT approximation algorithm for $\alpha$-edge-crossing width in Section~\ref{sec:edgecrossing} and discuss algorithmic applications in Section~\ref{sec:applications}. In Section~\ref{sec:relationship}, we establish the relationship between width parameters as presented in Figure~\ref{fig:parameters}.  We conclude and present some open problems in Section~\ref{sec:conclusion}.

\section{Preliminaries}\label{sec:prelim}
For a set $X$ and a positive integer $n$, we call ${X\choose n}$ the set of all subsets of $X$ of size exactly $n$. Let $\mathbb{N}$ be the set of all non-negative integers, and for a positive integer $n$, let $[n]=\{1,2,\cdots,n\}$.
For each function $F:A\to B$ and a subset $A'\subseteq A$, we denote by $F|_{A'}$ a function from $A'$ to $B$ satisfying $F|_{A'}(x)=F(x)$ for every $x\in A'$.
For two functions $F:A\to B$ and $G: C\to D$ with $A\cap C=\emptyset$, we consider $F\cup G$ as the function $W$ from $A\cup C$ to $B\cup D$ satisfying $W|_A=F$ and $W|_C=G$.

For a graph $G$, we denote by $V(G)$ and $E(G)$ the vertex set and the edge set of $G$, respectively.
Let $G$ be a graph. 
For a set $S$ of vertices in $G$, let $G[S]$ denote the subgraph of $G$ induced by $S$, and 
let $G-S$ denote the subgraph of $G$ obtained by removing all the vertices in $S$.
For $v\in V(G)$, let $G-v:=G-\{v\}$.
For an edge $e$ of $G$, let $G-e$ denote the graph obtained from $G$ by deleting $e$.
The set of neighbors of a vertex $v$ is denoted by $N_G(v)$, and the \emph{degree} of $v$ is the size of $N_G(v)$.
For two disjoint sets $S_1, S_2$ of vertices in $G$, we denote by $\delta_G(S_1,S_2)$ the set of edges incident with both $S_1$ and $S_2$ in $G$.

An edge $e$ of a connected graph $G$ is a \emph{cut edge} if $G-e$ is disconnected. A connected graph is \emph{2-edge-connected} if it has no cut edges.

For two graphs $G$ and $H$, we say that $G$ is a \emph{subdivision} of $H$ if $G$ can be obtained from $H$ by subsequently subdividing edges.

Let $T$ be a rooted tree with root node $r$. A path of $T$ is \emph{rooted} if it is a subpath of the path in $T$ from some node $t$ to $r$. For a rooted path $P$ in $T$, the node of $P$ that is closest to $r$ in $T$ is called the \emph{top node} of $P$, and the node of $P$ that is farthest from $r$ in $T$ is called the \emph{bottom node} of $P$.

Let $K_{1,n}$ be the star with $n$ leaves.

\subsection{Width parameters}
A \emph{tree-decomposition} of a graph $G$ is a pair $(T, \{B_t\}_{t\in V(T)})$ consisting of a tree $T$ and a family of sets $\{B_t\}_{t\in V(T)}$ of vertices in $G$ such that 
(1) $V(G)=\bigcup_{t\in V(T)}B_t$,
    (2) for every edge $uv$ of $G$, there exists a node $t$ of $T$ such that $u,v\in B_t$, and 
    (3) for every vertex $v$ of $G$, the set $\{t\in V(T):v\in B_t\}$ induces a subtree of $T$.
    The \emph{width} of a tree-decomposition is $\max_{t\in V(T)}|B_t|-1$, and the \emph{tree-width} of a graph, denoted by $\tw(G)$, is the maximum width over all its tree-decompositions.
    
    A tree-decomposition $(T,\{B_t\}_{t\in V(T)})$ is called \emph{rooted} if $T$ is a rooted tree. 
A rooted tree-decomposition $(T,\{B_t\}_{t\in V(T)})$ with a root $r$ is called \emph{nice} if the following hold;
\begin{itemize}
    \item For a non-root leaf $t$ of $T$, $\abs{B_t}=1$.
    \item If a node $t$ is not a non-root leaf of $T$, then it is one of the following;
    \begin{itemize}
        \item (Forget node) $t$ has only one child $t'$ and $B_{t}=B_{t'}\setminus\{v\}$ for some $v\in B_{t'}$.
        \item (Introduce node) $t$ has only one child $t'$ and $B_{t}=B_{t'}\cup\{v\}$ for some $v\in V(G)\setminus B_{t'}$.
        \item (Join node) $t$ has exactly two children $t_1$ and $t_2$ and $B_t=B_{t_1}=B_{t_2}$.
    \end{itemize}
\end{itemize}

We will use the following approximation algorithm due to Korhonen.
\begin{theorem}[Korhonen~\cite{TuukaTwalgo}]\label{twalgo}
    There is an algorithm running in $2^{\mathcal{O}(w)}n$ time, that given an $n$-vertex graph $G$ and an integer $w$,   either outputs a tree-decomposition of $G$ of width at most $2w+1$ or reports that the tree-width of $G$ is more than $w$.
\end{theorem}

By applying the following lemma, we can find a nice tree-decomposition.
\begin{lemma}[folklore; see Lemma 7.4 in~\cite{Book:CyganParaAlgo}]\label{lem:tdtonicetd}
    Given a tree-decomposition of an $n$-vertex graph $G$ of width $w$, one can construct a nice tree-decomposition $(T,\mathcal{B})$ of width $w$ with $\abs{V(T)}=\mathcal{O}(wn)$ in $\mathcal{O}(w^2\cdot\max(\abs{V(T)},n))$ time.
\end{lemma}

We prove the following relationship between tree-width and $\alpha$-edge-crossing width. This will be used in our approximation algorithm for $\alpha$-edge-crossing width in Section~\ref{sec:edgecrossing}, and also used to prove that edge-crossing width is equivalent to tree-partition-width in Subsection~\ref{subsec:ecrwtpw}.

\begin{lemma}\label{lem:twecrw}
    For every graph $G$ and every positive integer $\alpha$, $\tw(G)\le 5\ecrw(G)-1$ and $\tw(G)\le 3\ecrw_{\alpha}(G)+2\alpha-1$. 
\end{lemma}
\begin{proof}
    We first show the first inequality.
    Let $k=\ecrw(G)$.
    Let $\mathcal{T}=(T, \{X_t\}_{t\in V(T)})$ be a tree-cut decomposition of $G$ of edge-crossing width $\ecrw(G)$.
    We consider $\mathcal{T}$ as a rooted tree-cut decomposition with root node $r$.
    Let $\sigma:V(G)\to V(T)$ be the function where $v$ is contained in $X_{\sigma(v)}$.

    We construct a rooted tree-decomposition $(T, \{B_t\}_{t\in V(T)})$ as follows.
    For each node $t$ of $T$, 
    let $F_t$ be the set of edges $ab$ of $G$ satisfying that 
    either 
    \begin{itemize}
        \item (type 1) the path between $\sigma(a)$ and $\sigma(b)$ in $T$ contains $t$ as an internal node, or
        \item (type 2) the path between $\sigma(a)$ and $\sigma(b)$ in $T$ has length at least $1$, and contains $t$ as an end node, and the subtree of $T$ rooted at $t$ does not contain the end node of this path other than $t$.
    \end{itemize}
    Let $B_t$ be the union of $X_t$ and the set of vertices of $G$ incident with an edge in~$F_t$.

Since $\cross_\mathcal{T}(t)\le k$, the number of the edges of type 1 is at most $k$. So, because of this type, we put at most $2k$ vertices into $B_t$.
    If $t$ is the root node, then there is no edge of type 2. Assume that $t$ is not the root node, and let $t'$ be its parent.
     For type 2, $\sigma(a)=t$ or $\sigma(b)=t$, and either the vertex in $\{a,b\}$ that is not contained in $X_t$ is contained in $B_{t'}$ or $ab$ crosses $X_{t'}$.
     Since the number of edges $ab$ of type 2 crossing $X_{t'}$ is at most $k$, we may add at most $k$ vertices other than $X_t\cup X_{t'}$. As $\abs{X_t\cup X_{t'}}\le 2k$, in total, we have that $\abs{B_t}\le (2k)+(2k)+k=5k$.  

    We now verify that $(T, \{B_t\}_{t\in V(T)})$ is a  tree-decomposition. 
    Since $X_t\subseteq B_t$ for each $t\in V(T)$, every vertex of $G$ appears in some bag. 
    Let $ab\in E(G)$ and assume that there is no bag of $\mathcal{T}$ containing both $a$ and $b$.  If the path between $\sigma(a)$ and $\sigma(b)$ in $T$ contains some node $t$ of $T$ as an internal node, then by the construction, $B_t$ contains both $a$ and $b$. Assume that 
    there is no node in $T$ that is an internal node of the path between $\sigma(a)$ and $\sigma(b)$ in $T$. This means that $\sigma(a)$ is adjacent to $\sigma(b)$ in $T$. By symmetry, we assume that $\sigma(a)$ is the parent of $\sigma(b)$. Then $ab$ is an edge of type 2 for the node $t=\sigma(b)$, and thus, $\{a,b\}\subseteq B_{\sigma(b)}$.
    Thus, $(T, \{B_t\}_{t\in V(T)})$ satisfies the second condition.

    Lastly, to see that $(T, \{B_t\}_{t\in V(T)})$ satisfies the third condition, let $a\in V(G)$.
    For every vertex $b\in V(G)$ adjacent to $a$ in $G$, 
    let $P_{ab}$ be the path between $\sigma(a)$ and $\sigma(b)$ in $T$. 
    We added $a$ to $B_x$ for all $x\in V(P_{ab}-\{\sigma(a),\sigma(b)\})$. 
    Also, when $\sigma(b)$ is a descendant of $\sigma(a)$, we added $a$ to $B_{\sigma(b)}$ as well.
    Since $B_{\sigma(a)}$ contains $a$, the subtree of $T$ induced by the union of all $t$ where $a\in B_t$ is connected.
        This implies that $(T, \{B_t\}_{t\in V(T)})$ satisfies the third condition. 
    
    It is straightforward to verify the second inequality using the same argument, by replacing the inequality $\abs{X_t\cup X_{t'}}\le 2k$ with $\abs{X_t\cup X_{t'}}\le 2\alpha$. 
\end{proof}

We recall tree-cut decomposition and tree-cut width defined by Wollan~\cite{Wollan2015}.
A \emph{tree-cut decomposition} of a graph $G$ is a pair $\mathcal{T}=(T,\mathcal{X})$ such that $T$ is a tree and $\mathcal{X}$ is a family $\{X_t\}_{t\in V(T)}$ of disjoint sets of vertices in $G$ allowing empty sets where $\bigcup_{t\in V(T)}X_t=V(G)$. The elements of $\mathcal{X}$ are called bags.
For an edge $e=\{u,v\}$ of $T$, let $T_{e,u}$ and $T_{e,v}$ be two subtrees of $T-uv$ which contain $u$ and $v$, respectively. We define the \emph{adhesion} of an edge $uv$, denoted by $\adh_\mathcal{T}(uv)$, as the set of all edges incident with both $\bigcup_{t\in V(T_{uv,u})}X_t$ and $\bigcup_{t\in V(T_{uv,v})}X_t$ in $G$.

The \emph{torso} $H_t$ of $\mathcal{T}$ at a node $t$ is the graph obtained from $G$ as follows. If $|V(T)|=1$, then let $H_t=G$. Otherwise, let $T_1,T_2,\cdots, T_m$ be the connected components of $T-t$, and for each $i\in [m]$, let $Z_i=\bigcup_{v\in V(T_i)}X_v$. The torso $H_t$ at a node $t$ is obtained from $G$ by consolidating each vertex set $Z_i$ into a single vertex $z_i$. The operation of consolidating a vertex $Z$ into $z$ is to substitute $Z$ by $z$ in $G$, and for each edge $e$ between $Z$ and $v\in V(G)\setminus Z$, to add an edge $zv$ in the new graph. Note that this procedure may create multi-edges.
For a vertex $v$ of degree at most $2$ in a graph $F$, the operation of \emph{suppressing} $v$ is to remove $v$ and if $v$ has degree exactly $2$, add an edge between the neighbors of $v$ in $F$.
The \emph{$3$-center} of $H_t$, denoted by $\widetilde{H_t}$, is the unique graph obtained from $H_t$ by exhaustively suppressing vertices of degree at most $2$ in $H_t$ contained in $V(H_t)\setminus X_t$.

The \emph{tree-cut width} of $\mathcal{T}$ is defined as
\[\max \left(\max_{uv\in E(T)}|\adh_\mathcal{T}(uv)|, \max_{t\in V(T)}|V(\widetilde{H_t})|\right).\] 
The \emph{tree-cut width} of $G$, denoted by $\tcw(G)$, is the minimum tree-cut width over all tree-cut decompositions of $G$.

Ganian and Korchemna~\cite{Ganian2022} introduced the slim tree-cut width using the 2-center of a torso. The \emph{$2$-center} of a torso $H_t$, denoted by $\widehat{H}_t$, is the graph obtained from $H_t$ by removing vertices of degree $1$ in $H_t$ contained in $V(H_t)\setminus X_t$.
The \emph{slim tree-cut width} of a tree-cut decomposition $\mathcal{T}=(T,\mathcal{X})$ of a graph $G$ is defined as \[\max\left(\max_{uv\in E(T)}|\adh_\mathcal{T}(uv)|,\max_{t\in V(T)}|V(\widehat{H_t})|\right).\] The \emph{slim tree-cut width} of $G$, denoted by $\stcw(G)$, is the minimum slim tree-cut width  over all tree-cut decompositions of $G$.

Ding and Oporowski~\cite{tpw1996} introduced the notion of tree-partition-width. A \emph{tree-partition} of a graph $G$ is a tree-cut decomposition $(T,\{X_t\}_{t\in V(T)})$ such that for every $uv\in E(G)$, either 
\begin{itemize}
    \item there is a node $t$ of $T$ such that $X_t$ contains both $u$ and $v$, or 
    \item there is an edge $t_1t_2$ of $T$ such that one of $u$ and $v$ is contained in $X_{t_1}$ and the other is contained in $X_{t_2}$.
\end{itemize}
The \emph{tree-partition-width} of $G$, denoted by $\tpw(G)$, is the minimum thickness over all tree-partitions of $G$.

\section{An FPT approximation algorithm for $\alpha$-edge-crossing width}\label{sec:edgecrossing}

In this section, we present an FPT approximation algorithm for $\alpha$-edge-crossing width. We similarly follow the strategy to obtain a 2-approximation algorithm for tree-cut width designed by Kim et al~\cite{KimOP2018}. We formulate a new problem called \textsc{Constrained Star-Cut Decomposition}, which corresponds to decomposing a large leaf bag in a tree-cut decomposition, and we want to apply this subalgorithm recursively. By Lemma~\ref{lem:twecrw}, we can assume that a given graph admits a tree-decomposition of bounded width, and we design a dynamic programming to solve \textsc{Constrained Star-Cut Decomposition} on graphs of bounded tree-width. 

For a graph $G$, we call a function from $V(G)$ to $\mathbb{N}$ a (vertex) weight function. For a weight function $\gamma$ and a non-empty vertex subset $S\subseteq V(G)$, we define $\gamma(S):=\sum_{v\in S}\gamma(v)$ and $\gamma(\emptyset):=0$.

\medskip
\noindent
\fbox{\parbox{0.97\textwidth}{
	\textsc{Constrained Star-Cut Decomposition}\\
	\textbf{Input :} A graph $G$, two positive integers $\alpha$, $k$, and a weight function $\gamma:V(G)\to \mathbb{N}$ \\
	\textbf{Question :} Determine whether there is a tree-cut decomposition $\mathcal{T}=(T, \{X_t\}_{t\in V(T)})$ of $G$ such that 
	\begin{itemize}
	    \item $T$ is a star with center $t_c$ and it has at least one leaf, 
	    \item $\abs{X_{t_c}}\le \alpha$ and $\cross_\mathcal{T}(t_c)\le k$,
	    \item for each leaf $t$ of $T$, $\gamma(X_t)\le \alpha^2+2k$ and $\abs{\delta_G(X_t, X_{t_c})}\le \alpha^2+k$, and 
        \item there is no leaf $q$ of $T$ such that $X_q=V(G)$.
	\end{itemize}}}
\medskip

We will consider the following situation. Let $G$ be a graph and $S$ be a set of vertices in $G$ where $\abs{S}\ge \alpha+1$ and the number of edges between $S$ and $V(G)\setminus S$ is at most $2\alpha^2+4k$. 
Lemma~\ref{lem:useSCrecur} shows that in this situation, if $G$ has $\alpha$-edge-crossing width at most $k$, then $(G[S], \alpha, k, \gamma_S)$ is a Yes-instance of \textsc{Constraint Star-Cut Decomposition}, where $\gamma_S(v)=\abs{\delta_G(\{v\}, V(G)\setminus S)}$. In Lemma~\ref{AlgoSC}, we design an algorithm that solves \textsc{Constraint Star-Cut Decomposition} on graphs of bounded tree-width. Based on these two lemmas, we design an approximation algorithm for $\alpha$-edge-crossing width in Theorem~\ref{thm:alphaecrwalgo}.

\begin{lemma}\label{lem:useSCrecur}
Let $G$ be a graph, let $\alpha, k$ be positive integers, and let $S$ be a set of vertices in $G$. Assume that $\abs{S}\ge \alpha+1$ and $\abs{\delta_G(S, V(G)\setminus S)}\le 2\alpha^2+4k$.
For each vertex $v\in S$, let $\gamma_S(v)=\abs{\delta_G(\{v\}, V(G)\setminus S)}$.

If $\ecrw_{\alpha}(G)\le k$, then $(G[S], \alpha, k, \gamma_S)$ is a Yes-instance of \textsc{Constraint Star-Cut Decomposition}.
\end{lemma}
\begin{proof}
Let $\mathcal{T}=(T, \{X_t\}_{t\in V(T)})$ be a tree-cut decomposition of $G$ of thicknesses at most $\alpha$ and crossing number at most $k$.
For an edge $e=uv$ of $T$, let $T_{e,u}$ and $T_{e,v}$ be two subtrees of $T-uv$ which contain $u$ and $v$,
respectively.

We want to identify a node $t_c$ of $T$ that will correspond to the central node of the resulting star decomposition.
First, we define an extension $\gamma$ on $V(G)$ of the weight function $\gamma_S$ on $S$ as $\gamma(v)=\gamma_S(v)$ if $v\in S$ and $\gamma(v)=0$, otherwise.
We orient an edge $e=xy\in E(T)$ from $x$ to $y$ if the edge $e$ satisfies at least one of the following rules;
\begin{enumerate}[Rule 1.]
    \item $S\cap \left(\bigcup_{t\in V(T_{e,x})}X_t\right)=\varnothing$.
    \item $\gamma\big(\bigcup_{t\in V(T_{e,y})}X_t\big)>\alpha^2+2k$.
\end{enumerate}
Note that an edge may have no direction.

\begin{claim}\label{claim:onedirection}
Every edge has at most one direction.
\end{claim}
\begin{clproof}
Suppose that $e=xy$ has two different directions.

First assume that Rule 1 gives an orientation from $x$ to $y$, or from $y$ to $x$.
By symmetry, we assume that Rule 1 gives an orientation from $x$ to $y$.
Then Rule 1 cannot give an orientation from $y$ to $x$, since $S$ is not empty. Also, Rule 2 cannot give an orientation from $y$ to $x$, because Rule 1 guarantees that $\bigcup_{t\in V(T_{e,x})}X_t$ contains no vertex of $S$, and thus, the $\gamma$-value of this set is $0$.

Now, we assume that Rule 2 gives orientations from $x$ to $y$ and from $y$ to $x$.
Then we have that  $2\alpha^2+4k<\gamma\big(\bigcup_{t\in V(T_{e,x})}X_t\big)+\gamma\big(\bigcup_{t\in V(T_{e,y})}X_t\big)$. This  contradicts the fact that $\gamma\big(\bigcup_{t\in V(T_{e, x})}X_t\big)+\gamma\big(\bigcup_{t\in V(T_{e, y})}X_t\big)= \gamma(V(G))\le 2\alpha^2+4k$.
\end{clproof}

By Claim~\ref{claim:onedirection}, $T$ has a node whose incident edges have no direction or a direction to the node. Take such a node as a central node $t_c$.
Let $T_1,\ldots, T_m$ be the connected components of $T-t_c$ such that for every $i\in[m]$, $\big(\bigcup_{t\in V(T_i)}X_t\big) \cap S\neq \emptyset$.
Note that there is at least one such component, because $\abs{S}>\alpha$ and $\abs{X_{t_c}}\le \alpha$.

Let $(T',\{X'_t\}_{t\in V(T')})$ be a tree-cut decomposition of $G[S]$ such that 
\begin{itemize}
    \item $T'$ is a star with the central node $t_c$ and leaves $t_1,\ldots, t_m$,
    \item $X'_{t_c}=X_{t_c}\cap S$, and
    \item for every $i\in [m]$, $X'_{t_i}=\left(\bigcup_{t\in V(T_{i})}X_t\right)\cap S$.
\end{itemize}
We claim that $(T',\{X'_t\}_{t\in V(T')})$ satisfies the conditions of answer of \textsc{Constrained Star-Cut Decomposition}.
By the construction of decomposition, the first condition holds. Since $X'_{t_c}=X_{t_c}\cap S$ and the crossing number of $t_c$ in $(T,\{X_t\}_{t\in V(T)})$ is at most $k$, the second condition also holds.

Let $t_i$ be a leaf of $T'$ and let $V_i=\bigcup_{t\in V(T_{i})}X_t$.
By Rule 2 of  the orientation, we have that $\gamma(X'_{t_i})\le \alpha^2+2k$.  Let $t'$ be the node in $T_i$ that is adjacent to $t_c$ in $T$. Since $\abs{X_{t_c}}\le \alpha$ and $\abs{X_{t'}}\le \alpha$, we have $\abs{\delta_{G} (X_{t_c}, X_{t'})}\le \alpha^2$. Furthermore, because $\cross_{\mathcal{T}}(t')\le k$, we have $\abs{\delta_{G} (X_{t_c}, V_i\setminus X_{t'})}\le k$. Thus, we have 
\[ \abs{\delta_{G[S]}( X_{t_c}', X_{t_i}')}\le \abs{\delta_{G}(X_{t_c}, V_i)} \le \alpha^2+k. \]
This shows that the third condition holds.

We verify the last condition.
\begin{claim}
There is no leaf $q$ of $T'$ such that $X_q'=S$.
\end{claim}
\begin{clproof}
Suppose that there is a leaf $q$ of $T'$ such that $X_q'=S$. This means that there is no other leaf in $T'$ and $X_{t_c}'=\emptyset$. Let $T^*$ be the connected component of $T-t_c$ for which $\big(\bigcup_{t\in V(T^*)}X_t\big)\cap S=X_q'$. 
Then the edge of $T$ between $T^*$ and $t_c$ should be oriented towards $T^*$, because 
$\big(\bigcup_{t\in V(T)\setminus V(T^{*})}X_t\big)\cap S=\emptyset$.
This contradicts the choice of $t_c$.
\end{clproof}

This proves the lemma.
\end{proof}

We now devise an algorithm for \textsc{Constraint Star-Cut Decomposition} on graphs of bounded tree-width. 
To design a dynamic programming algorithm, we consider a pair consisting of a family $\mathcal{X}=(X_0, X_1, \ldots, X_{2k}, Y)$ of vertex sets and a partition $\mathcal{P}$ of $Y$, called a \emph{legitimate pair}, and we recursively enumerate all possible legitimate pairs.  Briefly speaking, a legitimate pair corresponds to a partial solution to \textsc{Constrained Star-Cut Decomposition}. The set $X_0$ corresponds to the center of the solution, and $X_1, \ldots, X_{2k}$ correspond to leaf bags incident with at most $k$ crossing edges, and parts of $\mathcal{P}$ correspond to leaf bags that do not incident with crossing edges. At the root node, if there is a legitimate pair, then it constitutes a solution to \textsc{Constrained Star-Cut Decomposition}.

\begin{lemma}\label{AlgoSC}
    Let $(G, \alpha, k, \gamma)$ be an instance of \textsc{Constrained Star-Cut Decomposition} and let $(T,\{B_t\}_{t\in V(T)})$ be a nice tree-decomposition of width at most $w$. In $2^{\mathcal{O}\left((k+w)\log(w(\alpha+k))\right)}\abs{V(T)}$ time, one can either output a solution of $(G, \alpha, k, \gamma)$, or correctly report that $(G, \alpha, k, \gamma)$ is a No-instance.
\end{lemma}
\begin{proof}
Let $r$ be the root of $T$. We design a dynamic programming to compute a solution of \textsc{Constrained Star-Cut Decomposition} in bottom-up sense. 
For each node $t$ of $T$, let $A_t$ be the union of all bags $B_{t'}$ where $t'$ is a descendant of $t$ in $T$.

Let $Z\subseteq V(G)$ be a set. 
A pair $(\mathcal{X}, \mathcal{P})$ of a sequence $\mathcal{X}=(X_0, X_1, \ldots, X_{2k}, Y)$ and a partition $\mathcal{P}$ of $Y$ is \emph{legitimate} with respect to $Z$  if  
\begin{itemize}
    \item $X_0, X_1, \ldots, X_{2k}, Y$ are pairwise disjoint subsets of $Z$ that are possibly empty, 
    \item $\left(\bigcup_{i\in \{0, 1, \ldots, 2k\}}X_i\right)\cup Y=Z$,
    \item $|X_0|\le \alpha$,
    \item for each $i\in [2k]$, $\abs{\delta_G(X_i,X_0)}\le \alpha^2+k$ and $\gamma(X_i)\le \alpha^2+2k$,
\item $\sum_{\{i,j\}\in {[2k]\choose 2}}\abs{\delta_{G}(X_i,X_j)}\le k$,
    \item for each $i\in [2k]$ and each $P\in \mathcal{P}$,  $\abs{\delta_G(X_0, P)}\le \alpha^2+k$, and
    $\abs{\delta_G(X_i,P)}=0$,
    \item for each $P\in \mathcal{P}$, $\gamma(P)\le \alpha^2+2k$,
    \item for any distinct sets $P_i, P_j\in \mathcal{P}$, $\abs{\delta_G(P_i, P_j)}=0$.
    \end{itemize}

\begin{claim}\label{claim:solution}
    $(G,\alpha,k,\gamma)$ is a Yes-instance if and only if there is a legitimate pair $(\mathcal{X}, \mathcal{P})$ with respect to $V(G)$ with $\mathcal{X}=(X_0, X_1, \ldots, X_{2k}, Y)$ such that any set of $X_1, \ldots, X_{2k}$ or a set of $\mathcal{P}$ is not the whole set $V(G)$.
\end{claim}
\begin{clproof}
    Assume that $(\mathcal{X}=(X_0, X_1, \ldots, X_{2k}, Y), \mathcal{P})$ is a legitimate pair with respect to $V(G)$ where any set of $X_1, \ldots, X_{2k}$ or a set of $\mathcal{P}$ is not the whole set $V(G)$. Let $(T,\{U_t\}_{t\in V(T)})$ such that $T$ is a star with center $t_0$ and leaves $t_1,t_2, \ldots, t_m$ where $m=2k+\abs{\mathcal{P}}$.
    For each integer $0\le i \le2k$, let $U_{t_i}:=X_i$. We assign parts of $\mathcal{P}$ as the bags $U_{t_{2k+1}}, \ldots, U_{t_m}$ if $\mathcal{P}$ is non-empty.
    Then $(T,\{U_t\}_{t\in V(T)})$ is a solution of \textsc{Constraint Star-Cut Decomposition}.
    
    Conversely, assume that $(G,\alpha,k,\gamma)$ is a Yes-instance and let $\mathcal{T}=(T,\{U_t\}_{t\in V(T)})$ be a solution for the instance. We may assume that $\mathcal{T}$ has no empty leaf bag.
    Let $t_c$ be the central node of $T$, and let $t_1, \ldots, t_x$ be the set of all node $t$ of $T$ such that there is at least one edge incident with $U_t$ and $V(G)\setminus (U_t\cup U_{t_c})$. As $\cross_{\mathcal{T}}(t_c)\le k$, we have $x\le 2k$. Let $t_{2k+1}, \ldots, t_m$ be the other nodes if one exists.

    Let $X_0=U_{t_c}$ and for all $1\le i\le x$, let $X_i=U_{t_i}$. If $x<2k$, then for all $x+1\le i\le 2k$, let $X_i=\emptyset$.
    Lastly, if there is a node $t_j$ with $j\ge 2k+1$, then let $Y=\bigcup_{2k+1\le j\le m}U_j$ and $\mathcal{P}=\{U_j:2k+1\le j\le m\}$. Otherwise, let $Y=\emptyset$ and $\mathcal{P}=\emptyset$.
    Observe that any set of $X_1, \ldots, X_{2k}$ or a set of $\mathcal{P}$ is not the whole set $V(G)$, because of the last condition of a solution.
    Therefore $((X_0,X_1,\ldots,X_{2k},Y),\mathcal{P})$ is a legitimate pair with respect to $V(G)$.
\end{clproof}

For a legitimate pair $(\mathcal{X},\mathcal{P})$ with respect to $Z$ and $Z'\subseteq Z$, let 
\begin{itemize}
\item $\mathcal{X}|_{Z'}=(X_0\cap Z',X_1\cap Z',\ldots, X_{2k}\cap Z',Y\cap Z')$ and 
\item $\mathcal{P}|_{Z'}=\{P\cap Z':P\in\mathcal{P}, P\cap Z'\neq \emptyset\}$.
\end{itemize}
We can see that  $(\mathcal{X}|_{Z'},\mathcal{P}|_{Z'})$ is legitimate with respect to $Z'$, because the constraints are the number of vertices in a set, the number of edges between two sets, and the sum of $\gamma$-values. Based on this fact, we will recursively store information about all legitimate pairs with respect to $A_t$ for nodes $t$.

Let $t$ be a node of $T$. A tuple $(I, \mathcal{Q}, C_1, C_2, D_1, D_2, a, b)$ is a \emph{valid tuple} at $t$ if
\begin{itemize}
    \item $I: B_t\to \{0,1,\ldots,2k, 2k+1\}$,
    \item $\mathcal{Q}$ is a partition of $I^{-1}(2k+1)$,
    \item $C_1:[2k]\to \{0, 1\ldots, \alpha^2+2k\}$, 
    \item $C_2:\mathcal{Q}\to \{0, 1\ldots, \alpha^2+2k\}$,
    \item $D_1:[2k]\to \{0, 1\ldots, \alpha^2+k\}$, 
    \item $D_2:\mathcal{Q}\to \{0, 1\ldots, \alpha^2+k\}$, and
    \item $a$, $b$ are two integers with $0\le a\le \alpha$ and $0\le b\le k$.
\end{itemize}
We say that a valid tuple $(I,\mathcal{Q},C_1,C_2,D_1,D_2,a,b)$ at a node $t$ represents a legitimate pair $(\mathcal{X},\mathcal{P})=((X_0, X_1, \ldots, X_{2k}, Y), \mathcal{P})$ with respect to $A_t$ if the following are satisfied;
\begin{itemize}
    \item for every $v\in B_t$, $I(v)=i$ if and only if $v\in\left\{ \begin{array}{ll}
    X_i & \textrm{if $0\le i\le 2k$}\\
    Y & \textrm{if $i=2k+1$}
    \end{array}\right.$
    \item $\mathcal{Q}=\mathcal{P}|_{B_t}$,
    \item for each $i\in[2k]$, $C_1(i)=\gamma(X_i)$,
    \item for each $Q\in \mathcal{Q}$, $C_2(Q)=\gamma(Q)$,
    \item for each $i\in[2k]$, $D_1(i)=\abs{\delta_G(X_i,X_0)}$,
    \item for each $Q\in \mathcal{Q}$, $D_2(Q)=\abs{\delta_G(Q,X_0)}$,
    \item $a=\abs{X_0}$,
    \item $b=\sum_{\{i,j\}\subseteq{[2k]\choose 2}}\abs{\delta_G(X_i,X_j)}$.
    \end{itemize}
We say that a valid tuple is a \emph{record} at $t$ if it represents some legitimate pair with respect to $A_t$.
Let $\mathcal{R}(t)$ be the set of all records at $t$. 

It is known that there is a constant $d$ such that the number of partitions of a set of $m$ elements is at most $dm^m$.
We define a function \[\zeta(x)=(2k+2)^{x}(d x^x)(\alpha^2+2k+1)^{2(x-1)+4k}(\alpha+1)(k+1).\] 
Observe that if $B_t$ has size $q$, then the number of all possible valid tuples at $t$ is at most $\zeta(q)$. Note that $\zeta(q)=2^{\mathcal{O}\left((q+k)\log(q(\alpha+k))\right)}$.

We describe how to store all records in $\mathcal{R}(t)$ for each node $t$ of $T$, depending on the type of $t$.

\medskip
\textbf{1) $t$ is a leaf node with $B_t=\{v\}$.}

We consider a sequence $\mathcal{X}=(X_0, X_1, \ldots, X_{2k}, Y)$ where one set is $\{v\}$ and all the others are empty, and a partition $\mathcal{P}$ of $Y$. Clearly, $(\mathcal{X}, \mathcal{P})$ is a legitimate pair with respect to $A_t$. 
For such a pair, we define $I, C_1, C_2, D_1, D_2, a, b$ such that 
\begin{itemize}
    \item $I(v)=\left\{ \begin{array}{ll}
    i & \textrm{if $v\in X_i$}\\
    2k+1 & \textrm{if $v\in Y$}
    \end{array} \right. $,
    \item for every $i\in [2k]$, $C_1(i)=\gamma(I^{-1}(i))$,
    \item for every $P\in \mathcal{P}$, $C_2(P)=\gamma(P)$, 
    \item for every $i\in [2k]$, $D_1(i)=0$, 
    \item for every $P\in \mathcal{P}$, $D_2(P)=0$,
    \item $a=\left\{ \begin{array}{ll}
    1 & \textrm{if $I(v)=0$}\\
    0 & \textrm{otherwise}
    \end{array} \right. $,
    \item $b=0$.
\end{itemize}
Then $(I,\mathcal{P},C_1,C_2,D_1,D_2,a,b)$ is a record that represents the legitimate pair $(\mathcal{X}, \mathcal{P})$. Let $\mathcal{R}(t)$ be the set of all records defined as above. By the construction, $\mathcal{R}(t)$ collects all the records representing some legitimate pair with respect to $A_t$.

In this case, we have $\abs{\mathcal{R}(t)}\le2k+2$.
So, $\mathcal{R}(t)$ is computed in $\mathcal{O}(k)$ time.

\medskip
\textbf{2) $t$ is a forget node with child $t'$ such that $B_t=B_{t'}\setminus \{v\}$.}

We construct a set $\mathcal{R}^*$ from $\mathcal{R}(t')$ as follows. 
Let $\mathcal{J}'=(I',\mathcal{Q}',C_1,C_2,D_1,D_2,a,b)$ be a record at $t'$. We construct a tuple $\mathcal{J}=(I,\mathcal{Q},C_1,C_2,D_1,D_2,a,b)$ as follows;
\begin{itemize}
    \item If $I'(v)\le 2k$, then we set $I=I'|_{B_{t}}$ and $\mathcal{Q}=\mathcal{Q}'$.
    \item If $I'(v)=2k+1$ and $v\in Q\in \mathcal{Q}'$, then we set $I=I'|_{B_{t}}$ and set $\mathcal{Q}$ to be the set obtained from $\mathcal{Q}'\setminus\{Q\}$ by adding $Q\setminus\{v\}$ if $Q\setminus\{v\}$ is non-empty.
\end{itemize}
Observe that $\mathcal{J}$ is a record at $t$.
 Let $(\mathcal{X}, \mathcal{P})$ be a legitimate pair with respect to $A_{t'}$ represented by $\mathcal{J}'$. 
Since $A_t=A_{t'}$, $(\mathcal{X}, \mathcal{P})$ is a legitimate pair with respect to $A_t$. As we obtained $\mathcal{J}$ from $\mathcal{J}'$ just by forgetting $v$ from $I'$ and $\mathcal{Q}'$, $\mathcal{J}$ represents the legitimate pair $(\mathcal{X}, \mathcal{P})$. This implies that $\mathcal{J}$ is a record at $t$. We add $\mathcal{J}$ to $\mathcal{R}^*$.

We claim that $\mathcal{R}^*=\mathcal{R}(t)$. By the above argument, $\mathcal{R}^*\subseteq \mathcal{R}(t)$. 
Let $\mathcal{J}\in\mathcal{R}(t)$. This means that $\mathcal{J}$ represents some legitimate pair $(\mathcal{X}, \mathcal{P})$ with respect to $A_t$.
Since $A_t=A_{t'}$, 
$(\mathcal{X}, \mathcal{P})$  is also a legitimate pair with respect to $A_{t'}$. So, a record $\mathcal{J}'$ representing $(\mathcal{X}, \mathcal{P})$ has been stored at $\mathcal{R}(t')$, and by the construction, $\mathcal{J}$ is computed from $\mathcal{J}'$. This shows that $\mathcal{R}(t)\subseteq \mathcal{R}^*$. 

Observe that $\abs{\mathcal{R}(t')}\le \zeta(w+1)$.
As we construct $\mathcal{R}(t)$ from $\mathcal{R}(t')$ by choosing a tuple in $\mathcal{R}(t')$  and modifying as above. This modification can be done in constant time. 
Thus, $\mathcal{R}(t)$ is computed in $2^{\mathcal{O}\left((w+k)\log(w(\alpha+k))\right)}$ time.

\medskip
\textbf{3) $t$ is an introduce node with child $t'$ such that $B_t=B_{t'}\cup\{v\}$.}

We construct a set $\mathcal{R}^*$ from $\mathcal{R}(t')$ as follows.
Let $\mathcal{J}'=(I',\mathcal{Q}',C'_1,C'_2,D'_1,D'_2,a',b')\in \mathcal{R}(t')$. 
For every $i\in \{0, 1, \ldots, 2k, 2k+1\}$ and $Q^*\in \mathcal{Q}'\cup \{\emptyset\}$ when $i=2k+1$, we construct a new tuple $\mathcal{J}=(I,\mathcal{Q},C_1,C_2,D_1,D_2,a,b)$ as follows:
\begin{itemize}
    \item $I(w)=\left\{ \begin{array}{ll}
    I'(w) & \textrm{if $w\in B_{t'}$}\\
    i & \textrm{if $w=v$}
    \end{array} \right. $,
    \item $\mathcal{Q}=\left\{ \begin{array}{ll}
    \mathcal{Q}'& \textrm{if $0\le I(v)\le 2k$}\\
   (\mathcal{Q}'\setminus\{Q^*\})\cup\{Q^*\cup\{v\}\} & \textrm{if $I(v)=2k+1$} 
    \end{array} \right. $,
     \item  for every $j\in[2k]$, 
    $C_1(j)=\left\{ \begin{array}{ll}
    C'_1(j) & \textrm{if $I(v)\neq j$}\\
    C'_1(j)+\gamma(v) & \textrm{if $I(v)=j$} 
    \end{array} \right. $,
    \item for every $Q\in \mathcal{Q}$, $C_2(Q)=\left\{ \begin{array}{ll}
    C'_2(Q) & \textrm{if $v\notin Q$}\\
    C'_2(Q\setminus \{v\})+\gamma(v) & \textrm{if $v\in Q$} 
    \end{array} \right. $,
    \item for every $j\in [2k]$, 
     $D_1(j)=\left\{ \begin{array}{ll}
    D'_1(j) & \textrm{if $I(v)\neq j$}\\
    D'_1(j)+\abs{\delta_G(\{v\},I^{-1}(0))} & \textrm{if $I(v)=j$} 
    \end{array} \right. $,
    \item for every $Q\in \mathcal{Q}$, $D_2(Q)=\left\{ \begin{array}{ll}
    D'_2(Q) & \textrm{if $v\notin Q$}\\
    D'_2(Q\setminus \{v\})+\abs{\delta_G(\{v\},I^{-1}(0))} & \textrm{if $v\in Q$} 
    \end{array} \right. $,
    \item $a=\left\{ \begin{array}{ll}
    a'+1 & \textrm{if $I(v)=0$}\\
    a' & \textrm{otherwise}
    \end{array}\right.$,
    \item $b=\left\{ \begin{array}{ll}
    b'+\abs{\delta_G\big(\{v\}, I^{-1}([2k]\setminus \{I(v)\})\big)  } & \textrm{if $1\le I(v)\le 2k$}\\
    b' & \textrm{otherwise}
    \end{array}\right.$.
\end{itemize}
We add this tuple to $\mathcal{R}^*$ whenever it is valid and 
\begin{itemize}
    \item if $1\le I(v)\le 2k$, then there is no edge between $v$ and $I^{-1}(2k+1)$ in $G$,
    \item if $I(v)=2k+1$, then there is no edge between $v$ and $I^{-1}(\{1, \ldots, 2k\})$ in $G$ and there is no edge between $v$ and $I^{-1}(2k+1)\setminus (Q^*\cup \{v\})$ in $G$.
\end{itemize}

We claim that $\mathcal{R}^*=\mathcal{R}(t)$. 

First we show that $\mathcal{R}^*\subseteq \mathcal{R}(t)$. Let  $\mathcal{J}$ be a valid tuple constructed as above from $\mathcal{J}'\in \mathcal{R}(t')$.
We have to show that $\mathcal{J}$ represents some legitimate pair with respect to $A_t$. Since $\mathcal{J}'\in \mathcal{R}(t')$, it represents a legitimate pair $(\mathcal{X}'=(X_0', \ldots, X_{2k}', Y'), \mathcal{P}')$ with respect to $A_{t'}$. 

If $0\le i\le 2k$, then we obtain $\mathcal{X}$ from $\mathcal{X}'$ by replacing $X_i'$ with $X_i'\cup \{v\}$ and set $\mathcal{P}=\mathcal{P}'$.
If $i=2k+1$, then we obtain $\mathcal{X}$ from $\mathcal{X}'$ by replacing $Y'$ with $Y'\cup \{v\}$ and 
\begin{itemize}
    \item adding $v$ to the part $P^*\in \mathcal{P}'$ with $P^*\cap B_t=Q^*$ when $Q^*\in \mathcal{Q}'$ or 
    \item adding a single part $\{v\}$ when $Q^*=\emptyset$,
    \end{itemize}
    to obtain a new partition $\mathcal{P}$ of $Y'\cup \{v\}$.
Then it is straightforward to verify that $(\mathcal{X}, \mathcal{P})$ is a legitimate pair with respect to $A_t$ and $\mathcal{J}$ represents it. This shows that $\mathcal{J}\in \mathcal{R}(t)$.

Now, we show that $\mathcal{R}(t)\subseteq\mathcal{R}^*$. Suppose that there is a record $\mathcal{J}=(I,\mathcal{Q},C_1,C_2,D_1,D_2,a,b)\in\mathcal{R}(t)$.
Then there is a legitimate pair $(\mathcal{X},\mathcal{P})$ represented by $\mathcal{J}$. Since a pair $(\mathcal{X}|_{B_{t'}},\mathcal{P}|_{B_{t'}})$ is legitimate, there is a record $\mathcal{J}'\in\mathcal{R}(t')$ which represents the pair. By the construction, $\mathcal{J}$ is computed from $\mathcal{J}'$. Hence $\mathcal{R}(t)\subseteq\mathcal{R}^*$.

In this case, we take one record in $\mathcal{R}(t')$ and construct a new tuple as explained above. After then, we check its validity. Note that $\abs{\mathcal{R}(t')}\le \zeta(w)$.
Checking the validity takes time $\mathcal{O}(w+k)$.
Hence, $\mathcal{R}(t)$ is computed in  $2^{\mathcal{O}((w+k)\log(w(\alpha+k)))}$ time.

\medskip
\textbf{4) $t$ is a join node with two children $t_1$ and $t_2$.}

We construct a set $\mathcal{R}^*$ from $\mathcal{R}(t_1)$ and $\mathcal{R}(t_2)$ as follows.

For two records $\mathcal{J}'=(I',\mathcal{Q}',C'_1,C'_2,D'_1,D'_2,a',b')\in \mathcal{R}(t_1)$ and $\mathcal{J}''=(I'',\mathcal{Q}'',C''_1,C''_2,D''_1,D''_2,a'',b'')\in \mathcal{R}(t_2)$ with $I'=I''$ and $\mathcal{Q}'=\mathcal{Q}''$, we construct a new tuple $\mathcal{J}=(I,\mathcal{Q},C_1,C_2,D_1,D_2,a,b)$ at $t$ as follows;
\begin{itemize}
    \item $I=I'=I''$,
    \item $\mathcal{Q}=\mathcal{Q}'=\mathcal{Q}''$,
    \item for all $i\in[2k]$, $C_1(i)=C'_1(i)+C''_1(i)-\gamma(I^{-1}(i))$,
    \item for all $Q\in\mathcal{Q}$, $C_2(Q)=C'_2(Q)+C''_2(Q)-\gamma(Q)$,
    \item  for all $i\in[2k]$, $D_1(i)=D'_1(i)+D''_1(i)-\abs{\delta_G(I^{-1}(i), I^{-1}(0))}$,
    \item for all $Q\in\mathcal{Q}$, $D_2(Q)=D'_2(Q)+D''_2(Q)-\abs{\delta_G(Q, I^{-1}(0))}$,
        \item $a=a'+a''-I^{-1}(0)$,
    \item $b=b'+b''-\sum_{\{i,j\}\in{[2k]\choose 2}}\abs{\delta_G(I^{-1}(i), I^{-1}(j))}$.
\end{itemize}
We add a new tuple to $\mathcal{R}^*$ when it is valid. 

We claim that $\mathcal{R}^*=\mathcal{R}(t)$. 

First show that $\mathcal{R}^*\subseteq \mathcal{R}(t)$.
Let $\mathcal{J}$ be a valid tuple constructed from $\mathcal{J}'$ and $\mathcal{J}''$ as above. Let $(\mathcal{X}'=(X_0', \ldots, X_{2k}', Y'), \mathcal{P}')$ and $(\mathcal{X}''=(X_0'', \ldots, X_{2k}'', Y''), \mathcal{P}'')$ be legitimate pairs with respect to $A_{t_1}$ and $A_{t_2}$, respectively. Note that since $\mathcal{Q}'=\mathcal{Q}''$, parts $P_1\in \mathcal{P}'$ and $P_2\in \mathcal{P}''$ intersect if and only if $P_1\cap B_t=P_2\cap B_t$. Let $\mathcal{P}$ be the family obtained from the disjoint union of $\mathcal{P}'$ and $\mathcal{P}''$ by merging two sets if they have a common vertex.  
Then $(\mathcal{X}, \mathcal{P})$ becomes a legitimate pair if the third to seventh conditions for being a legitimate pair hold. These conditions clearly hold when the tuple is valid. So, $\mathcal{J}\in \mathcal{R}(t)$.

Now, we show that $\mathcal{R}(t)\subseteq\mathcal{R}^*$. Suppose that there is a record $\mathcal{J}=(I,\mathcal{Q},C_1,C_2,D_1,D_2,a,b)\in\mathcal{R}(t)$. Since $\mathcal{J}$ is a record at $t$ and $B_t=B_{t_1}=B_{t_2}$, there are two records $(I',\mathcal{Q}',C'_1,C'_2,D'_1,D'_2,a',b')\in\mathcal{R}(t_1)$ and $(I'',\mathcal{Q}'',C''_1,C''_2,D''_1,D''_2,a'',b'')\in\mathcal{R}(t_2)$ with $I=I'=I''$ and $\mathcal{Q}=\mathcal{Q}'=\mathcal{Q}''$. By the constructions in each type of nodes, they are unique.

Note that $\mathcal{R}(t_1)$ and $\mathcal{R}(t_2)$ have at most $\zeta(w+1)$ tuples.
Merging two tuples takes time $\mathcal{O}(w+k)$, and checking the validity takes time $\mathcal{O}(w+k)$.
Therefore, $\mathcal{R}(t)$ is computed in $2^{\mathcal{O}\left((w+k)\log(w(\alpha+k))\right)}$ time.  

\medskip
Overall, this algorithm runs in $2^{\mathcal{O}\left((w+k)\log(w(\alpha+k))\right)}\abs{V(T)}$ time.
\end{proof}

Now, we provide an approximation algorithm for $\alpha$-edge-crossing width.

\begin{theorem}\label{thm:alphaecrwalgo}
    Given an $n$-vertex graph $G$ and two positive integers $\alpha$ and $k$, one can in time $2^{\mathcal{O}\left((\alpha+k)\log(\alpha+k)\right)}n^2$ either 
    \begin{itemize}
        \item output a tree-cut decomposition of $G$ with thickness at most $\alpha$ and crossing number at most $2\alpha^2+5k$, or
        \item correctly report that $\ecrw_\alpha(G)>k$.
    \end{itemize}
\end{theorem}

\begin{proof}
We recursively apply the algorithm for \textsc{Constrained Star-Cut Decomposition} as follows. At the beginning, we consider a trivial tree-cut decomposition with one bag containing all the vertices. In the recursive steps, we assume that we have a tree-cut decomposition $\mathcal{T}=(T, \{X_t\}_{t\in V(T)})$ such that  
\begin{enumerate}[(i)]
    \item for every internal node $t$ of $T$, $\abs{X_t}\le \alpha$ and $\cross_{\mathcal{T}}(t)\le 2\alpha^2+5k$, 
    \item for every leaf node $t$ of $T$, $\abs{\delta_G(X_t, V(G)\setminus X_t)}\le 2\alpha^2+4k$.
\end{enumerate}
If all leaf bags have size at most $\alpha$, then this decomposition has thickness at most $\alpha$ and crossing number at most $2\alpha^2+5k$. Thus, we may assume that there is a leaf bag $X_\ell$ having at least $\alpha+1$ vertices.

We apply  Theorem~\ref{twalgo} for $G[X_{\ell}]$ with $w=3k+2\alpha-1$. Then in time $2^{\mathcal{O}(k+\alpha)}n$, either we have a tree-decomposition of width at most $2(3k+2\alpha-1)+1=6k+4\alpha-1$ or we report that $\tw(G)\ge \tw(G[X_\ell])>3k+2\alpha-1$.
In the latter case, by Lemma~\ref{lem:twecrw}, we have $\ecrw_\alpha(G)>k$. Thus, we may assume that we have a tree-decomposition of $G[X_{\ell}]$ of width at most $6k+4\alpha-1$. By applying Lemma~\ref{lem:tdtonicetd}, we can find a nice tree-decomposition $(F, \{B_t\}_{t\in V(F)})$ of $G[X_{\ell}]$ of width at most $6k+4\alpha-1$ with $\abs{V(F)}=\mathcal{O}((k+\alpha)n)$ in time $\mathcal{O}((k+\alpha)^3 n)$. 

We define $\gamma$ on $X_{\ell}$ so that $\gamma(v)=\abs{\delta_G(\{v\}, V(G)\setminus X_\ell)}$. We run the algorithm in Lemma~\ref{AlgoSC} for the instance $(G[X_\ell], \alpha, k, \gamma)$. Then in time $2^{\mathcal{O}\left((\alpha+k)\log(\alpha+k)\right)}\abs{V(F)}$, one can either output a solution of $(G[X_{\ell}], \alpha, k, \gamma)$, or correctly report that $(G[X_{\ell}], \alpha, k, \gamma)$ is a No-instance. In the latter case, by Lemma~\ref{lem:useSCrecur}, we have $\ecrw_\alpha(G)>k$. 
In the former case, 
let $\mathcal{T}^*=(T^*, \{Y_t\}_{t\in V(T^*)})$ be the outcome, where $q_c$ is the center of $T^*$ and $q_1, \ldots, q_m$ are the leaves of $T^*$.
Then we modify the tree-cut decomposition $\mathcal{T}$ by replacing $X_{\ell}$ with $Y_{q_c}$ and then attaching bags $Y_{q_i}$ to $Y_{q_c}$, where corresponding nodes are $q_c$ and $q_1, \ldots, q_m$. 
Let $\mathcal{T}'$ be the resulting tree-cut decomposition. 

Observe that $\abs{Y_{q_c}}\le \alpha$ and $\cross_{\mathcal{T}^*}(q_c)\le k$. So, we have $\cross_{\mathcal{T}'}(q_c)\le k+(2\alpha^2+4k)=2\alpha^2+5k$. 
Also, for each $i\in [m]$, we have 
\begin{align*}
\abs{\delta_G(Y_{q_i}, V(G)\setminus Y_{q_i})}&\le \abs{\delta_G(Y_{q_i}, V(G)\setminus X_{\ell})}+ \abs{\delta_G(Y_{q_i}, Y_{q_c})}+\cross_{\mathcal{T}^*}(q_c) \\
&\le \gamma(Y_{q_i})+ \abs{\delta_G(Y_{q_i}, Y_{q_c})}+\cross_{\mathcal{T}^*}(q_c) \\
&\le (\alpha^2+2k)+(\alpha^2+k)+k=2\alpha^2+4k. 
\end{align*}
Therefore, we obtain a refined tree-cut decomposition with properties (i) and (ii). Note that by the last condition of the solution for \textsc{Constraint Star-Cut Decomposition}, new leaf bags have size less than $X_{\ell}$.
Thus, the algorithm will terminate in at most $n$ recursive steps.
When this procedure terminates, we either 
obtain a tree-cut decomposition of $G$ of thickness at most $\alpha$ and crossing number at most $2\alpha^2+5k$ or,
    conclude that $\ecrw_\alpha(G)>k$.
    
    The total running time is $(2^{\mathcal{O}\left((\alpha+k)\log(\alpha+k)\right)}n)\cdot n=2^{\mathcal{O}\left((\alpha+k)\log(\alpha+k)\right)}n^2$.
\end{proof}

\section{Algorithmic applications on coloring problems}\label{sec:applications}

In this section, we prove Theorem~\ref{thm:listcoloring}, which is split into Theorem~\ref{thm:lcalpha} and Corollary~\ref{cor:precoloring}.
In Theorem~\ref{thm:lcalpha}, we prove that \textsc{List Coloring} is fixed-parameter tractable when parameterized by $\alpha$-edge-crossing width, for every fixed $\alpha$. In Corollary~\ref{cor:precoloring}, we show that \textsc{Precoloring Extension} is also fixed parameter tractable parameterized by $\alpha$-edge-crossing width.

A vertex-coloring $f:V(G)\to \mathbb{N}$ on a graph $G$ is said to be \emph{proper} if $f(u)\neq f(v)$ for all edges $uv\in E(G)$. For a given set $\{L(v)\subseteq \mathbb{N}:v\in V(G)\}$, a coloring $c:V(G)\to \mathbb{N}$ is called an \emph{$L$-coloring} if $c(v)\in L(v)$ for all $v\in V(G)$.

\noindent
\fbox{\parbox{0.97\textwidth}{
	\textsc{List Coloring}\\
	\textbf{Input :} A graph $G$ and a set of lists $\mathcal{L}=\{L(v)\subseteq \mathbb{N}:v\in V(G)\}$
	\\
	\textbf{Question :} Does $G$ admit a proper $L$-coloring $c:V(G)\to\bigcup\mathcal{L}$? }}
\medskip

\noindent
\fbox{\parbox{0.97\textwidth}{
	\textsc{Precoloring Extension}\\
	\textbf{Input :} A graph $G$, a subset $S$ of $V(G)$, a positive integer $q$, and a proper coloring $c_S$ from $G[S]$ to $[q]$
	\\
	\textbf{Question :} Does $G$ admit a proper coloring $c:V(G)\to [q]$ with $c(v)=c_S(v)$ for all $v\in S$? }}
\medskip

We provide a sketch of the proof for Theorem~\ref{thm:lcalpha}. 

Assume that a tree-cut decomposition $(T, \mathcal{X})$ of the input graph $G$ of thickness at most $\alpha$ and crossing number $w$ is given.  
Note that we can obtain a decomposition of width $w=2\alpha^2+5\ecrw_\alpha(G)$ using Theorem~\ref{thm:approxalpha}. We consider it as a rooted decomposition. Let $t$ be a node of $T$, and let $G_t$ be the graph induced by the union of all $X_{t'}$ where $t'$ is a descendant of $t$. As $(T, \mathcal{X})$ has crossing number $w$, there are at most $w+\alpha$ vertices in $G_t$ having a neighbor in $V(G)\setminus V(G_t)$. By symmetry, also, there are at  most $w+\alpha$ vertices in $V(G)\setminus V(G_t)$ having a neighbor in $V(G_t)$.

    If we store all possible remaining colorings on these $w+\alpha$ boundaried vertices, then in the worst case, we need to store $n^{w+\alpha}$ many colorings, which is not helpful to obtain a fixed parameter algorithm.  We show  in Lemma~\ref{lem:compatible} that, using a representative set technique, the number of colorings to store can be reduced to $g(w+\alpha)$ for some function $g$, which does not depend on the size of the graph. This is one of the key ideas.

    Now, let $t$ be a node of $T$ and let $t_1, \ldots, t_m$ be its children. Because $(T, \mathcal{X})$ has crossing number $w$, there is a set $I_1$ of at most $2w$ integers $i$ in $\{1, 2, \ldots, m\}$ for which 
    there is an edge that is incident with both $V(G_{t_i})$ and $V(G)\setminus V(G_{t_i})\setminus X_t$. Let $I_2=\{1, 2, \ldots, m\}\setminus I_1$. For $j\in I_2$, we know that the set of neighbors of vertices in $V(G_{t_j})$ are all contained in $V(G_{t_j})\cup X_t$.  
    
    We first enumerate all possible colorings on $X_t$ in time $\mathcal{O}(n^\alpha)$, and then compare with stored colors for $G_{t_j}$ with $j\in I_2$. We only remain a coloring of $X_t$ such that for every $j\in I_2$, there is at least one coloring of $G_{t_j}$ that is compatible with this coloring on $G_{t_j}$. Once we have done this process, we can forget about graphs $G_{t_j}$ with $j\in I_2$. 
    
    In the next step, we need to consider graphs $G_{t_i}$ with $i\in I_1$.
    We consider all possible combinations of a coloring in $X_t$ and stored colorings for $G_{t_j}$ with $j\in I_1$, and remain only compatible colorings. Since the size of $I_1$ is at most $2w$ and in each $G_{t_i}$ we stored $g(w+\alpha)$ many colorings, this process will be done in FPT time. At the end, we apply the representative set technique for remaining colorings to reduce the total number of colorings on boundaried vertices.

\medskip
 We prove lemmas concerning the representative set technique.
For $V\in \mathbb{N}^q$, $W\in \mathbb{N}^t$, and $B\subseteq [q]\times [t]$,
we say that $(V, W)$ is \emph{$B$-compatible} if $V[i]\neq W[j]$ for every $(i, j)\in B$. When we have a vertex partition $(X, Y)$ of a graph $G$, possible colorings on boundaried vertices in $X$ and $Y$ will be related to vectors $V$ and $W$. 

\begin{lemma}\label{lem:compatible}
 Let $q$ and $t$ be positive integers,  and let $B\subseteq [q]\times [t]$.
 For every set $\mathcal{P}$ of distinct vectors in $\mathbb{N}^q$,
 there is a subset $\mathcal{P}^*$ of $\mathcal{P}$ of size at most $2^{\frac{q(q+1)}{2}}t^{q-1}(t+1)$ satisfying that  
 \begin{itemize}
     \item[($\ast$)] for every $W\in \mathbb{N}^{t}$, if there is $V\in \mathcal{P}$ where $(V, W)$ is $B$-compatible, then there is $V^*\in \mathcal{P}^*$ where $(V^*, W)$ is $B$-compatible.
 \end{itemize}
 Furthermore, such a set $\mathcal{P}^*$ can be computed in time $\mathcal{O}(\abs{\mathcal{P}} 2^{q^2} t^{q+2})$.
\end{lemma}
\begin{proof}
Let $g(1, t)=t+1$ and for $q>1$, let 
\[g(q, t)=t \cdot 2^q \cdot g(q-1, t). \]
Observe that $g(q, t)\le 2^{\frac{q(q+1)}{2}}t^{q-1}(t+1)$.

We prove by induction on $q$ that for every set $\mathcal{P}$ of distinct vectors in $\mathbb{N}^q$,
there is a subset $\mathcal{P}^*$ of $\mathcal{P}$ of size at most $g(q, t)$ satisfying $(\ast)$. This will prove the lemma.

First assume that $q=1$. If $\abs{\mathcal{P}}\ge t+1$, then let $\mathcal{P}^*$ be any subset of $\mathcal{P}$ of size $t+1$ and otherwise, let $\mathcal{P}^*=\mathcal{P}$. 
When $\mathcal{P}^*=\mathcal{P}$, $(\ast)$ is clearly satisfied. Assume that $\abs{\mathcal{P}}\ge t+1$ and there are $W\in \mathbb{N}^t$ and $V\in \mathcal{P}$ such that $(V, W)$ is $B$-compatible. As $\abs{\mathcal{P}^*}=t+1$, there is $V^*\in \mathcal{P}^*$ whose element does not appear in $W$. So, $(V^*, W)$ is $B$-compatible.  Thus, the condition $(\ast)$ is satisfied.

Now, we assume that $q>1$.

Let $S=\{(i,i):i\in [q]\}$.
We first greedily find a maximal set $\mathcal{B}$ of pairwise $S$-compatible elements in $\mathcal{P}$, that is up to the size $t+1$. We can construct such a set in time $\mathcal{O}(\abs{\mathcal{P}} q^2t^2)$. If $\abs{\mathcal{B}}=t+1$, then for every $W\in \mathbb{N}^t$, we can always find an element $V^*\in \mathcal{B}$ that is $B$-compatible with $W$. So, we can set $\mathcal{P}^*=\mathcal{B}$. Thus, we may assume that $\abs{\mathcal{B}}\le t$. As $\mathcal{B}$ is maximal, we have the property that 
\begin{itemize}
    \item for every $V\in \mathcal{P}\setminus \mathcal{B}$, there exists $B\in \mathcal{B}$ where $V[i]=B[i]$ for some $i\in [q]$. 
\end{itemize}

Now, for each $B\in \mathcal{B}$ and each non-empty set $I\subseteq [q]$, we define $C[B, I]$ as the set of all vectors $Z$ in $\mathcal{P}\setminus \mathcal{B}$ such that 
\begin{itemize}
    \item for all $i\in I$, $B[i]=Z[i]$, and 
    \item for all $i\in [q]\setminus I$, $B[i]\neq Z[i]$.
\end{itemize}

For each set $C[B, I]$, 
we compute a subset $C^*[B, I]$ satisfying the property $(\ast)$ for restrictions on vectors on the coordinates in $[q]\setminus I$. That is, we obtain a set by applying the induction hypothesis to $\{V|_{[q]\setminus I}:V\in C[B, I]\}$, and take the corresponding subset of $C[B, I]$. By induction, each set $C^*[B, I]$ has size at most $g(q-\abs{I},t)\le g(q-1,t)$.
Let 
\[\mathcal{P}^*=\mathcal{B}\cup \left(\bigcup_{B\in \mathcal{B}, \emptyset\subsetneq I\subseteq [q]} C^*[B, I]\right).\]

We claim that $\mathcal{P}^*$ satisfies the property $(\ast)$.
Note that 
\[\abs{\mathcal{P}^*}\le t+ t \cdot (2^q-1) \cdot g(q-1, t)\le t \cdot 2^q \cdot g(q-1, t).   \]

Let $W\in \mathbb{N}^t$ and assume that there is $V\in \mathcal{P}$ where $(V, W)$ is $B$-compatible. 
If $V$ is in $\mathcal{P}^*$, then there is nothing to prove. Thus, we may assume that $V\notin \mathcal{P}^*$. It means that $V\in C[B,I]\setminus C^*[B, I]$ for some $B\in \mathcal{B}$ and $\emptyset\neq I\subseteq [q]$.
By the construction of $C^*[B, I]$, there is $V^*\in C^*[B, I]$ where $(V^*|_{[q]\setminus I}, W)$ is $\big(B\cap (([q]\setminus I)\times [t])\big)$-compatible. As $V|_I=V^*|_I$, the pair $(V^*|_I, W)$ is also  $\big(B\cap (I\times [t])\big)$-compatible. Therefore, $(V^*, W)$ is $B$-compatible.

Now, we analyze the running time. When $q=1$, the set $\mathcal{P}^*$ is computed in time $\mathcal{O}(\abs{\mathcal{P}})$.  For $q>1$, we first compute $\mathcal{B}$ in time $\mathcal{O}(\abs{\mathcal{P}}qt^2)$, and then for every $B\in \mathcal{B}$ and non-empty set $I\subseteq [q]$, we apply the induction hypothesis with smaller $q$. It implies that the whole algorithm runs in time $\mathcal{O}(\abs{\mathcal{P}}\cdot (qt^2) \cdot (2^q qt)^q))=\mathcal{O}(\abs{\mathcal{P}} 2^{q^2} t^{q+2})$.
\end{proof}

Let $G$ be a graph. For disjoint sets $S, T$ of vertices in $G$ and functions $g:S\to \mathbb{N}$ and $h:T\to \mathbb{N}$, we say that $(S, g)$ is \emph{compatible} with $(T, h)$ if for every edge $vw$ with $v\in S$ and $w\in T$, $g(v)\neq h(w)$. If $S$ and $T$ are clear from the context, we simply say that $g$ and $h$ are compatible.

Using Lemma~\ref{lem:compatible}, we can show the following.
\begin{lemma}\label{lem:compgrph}
    Let $q$ and $t$ be positive integers. 
    Let $G$ be a graph and let $S$ and $T$ be disjoint sets of vertices in $G$ such that $\abs{\{v\in S:N_G(v)\cap T\neq\emptyset\}}\le q$ and $\abs{\{v\in T:N_G(v)\cap S\neq\emptyset\}}\le t$.
    Then for every set $\mathcal{F}$ of proper colorings of $G[S]$, there is a subset $\mathcal{F}^*$ of $\mathcal{F}$ of size at most $2^{\frac{q(q+1)}{2}}t^{q-1}(t+1)$ such that \begin{itemize}
        \item for every proper coloring $h$ of $G[T]$, if there is a pair $(S, g)$ with $g\in \mathcal{F}$ that is compatible with $(T, h)$, then there is a pair $(S, g^*)$ with $g^*\in \mathcal{F}^*$ that is compatible with $(T, h)$. 
    \end{itemize}
     Furthermore, such a set $\mathcal{F}^*$ can be computed in time $\mathcal{O}(\abs{\mathcal{F}} 2^{q^2} t^{q+2})$.
\end{lemma}
\begin{proof}
    Let $bd(S)=\{v\in S:N_G(v)\cap T\neq\emptyset\}=\{s_1, s_2, \ldots, s_x\}$ and $bd(T)=\{v\in T:N_G(v)\cap S\neq\emptyset\}=\{t_1, t_2, \ldots, t_y\}$. Observe that $(S, g)$ is compatible with $(T, h)$ if and only if $(bd(S), g|_{bd(S)})$ is compatible with $(bd(T), h|_{bd(T)})$.
    
    For each function $g:S\to \mathbb{N}$, we consider a vector $A_g=(g(s_1), g(s_2), \ldots, g(s_x))$, and similarly, for each function $h:T\to \mathbb{N}$, we consider a vector $B_h=(h(t_1), h(t_2), \ldots, h(t_y))$. Note that 
    $(bd(S), g|_{bd(S)})$ is compatible with $(bd(T), h|_{bd(T)})$ if and only if 
    $(A_g, B_h)$ is $U$-compatible where $U=\{(i, j):s_it_j\in E(G), s_i\in S, t_j\in T\}$.
    Thus, we can apply Lemma~\ref{lem:compatible} to the set of vectors $A_g$ for $g\in \mathcal{F}$ to obtain the required set $\mathcal{F}^*$.
\end{proof}

\begin{theorem}\label{thm:lcalpha}
    For a fixed positive integer $\alpha$, the \textsc{List Coloring} problem is fixed parameter tractable parameterized by $\alpha$-edge-crossing width.
\end{theorem}
\begin{proof}
We describe the algorithm for connected graphs. 
If a given graph is disconnected, then we can apply the algorithm for each component.
We assume that a given graph $G$ is connected.

Assume that there is a vertex $v$ where the list $L(v)$ has size more than the degree of $v$. Then after having a proper $L$-coloring $G-v$, we can extend it to a proper $L$-coloring of $G$ by selecting a color in $L(v)$ that does not appear in the neighborhood of $v$. Since $\ecrw_\alpha(G-v)\le \ecrw_\alpha(G)$, we can apply the algorithm for $G-v$ and then extend to $G$.
Therefore, we may assume that 
each list $L(v)$ has size at most the degree of $v$.

Let $\ecrw_\alpha(G)=k$.
Using the algorithm in Theorem~\ref{thm:alphaecrwalgo}, we obtain a tree-cut decomposition $\mathcal{T}=(T,\mathcal{X}=\{X_t\}_{t\in V(T)})$ of the input graph $G$ of thickness at most $\alpha$ and crossing number $w\le 2\alpha^2+5k$.
We consider it as a rooted decomposition by choosing a root node $r$ with $X_r\neq\emptyset$.

For every node $t\in V(T)$, we denote by $T_t$ the subtree of $T$ rooted at $t$, and let $G_t=G[\bigcup_{v\in V(T_t)}X_v]$.
For every node $t\in V(T)$, the \emph{boundary} $\partial(t)$ of $T_t$ is a graph $H$ where $E(H)$ is the set of edges incident with both $V(G_t)$ and $V(G)\setminus V(G_t)$, and $V(H)$ is the set of vertices in $G$ incident with an edge in $E(H)$. 
One can observe that $|E(\partial(t))|\le \alpha^2+2w$ for any $t\in V(T_t)$.
Let $\hat{\partial}(t):= (V(\partial(t))\cap V(G_t))\cup X_t$.
Note that $\abs{\hat{\partial}(t)}\le \alpha+w$  for any node $t$ of $T$.

Let $t\in V(T)$. A coloring $g$ on $\hat{\partial}(t)$ is \emph{valid at $t$} if there is a proper $L$-coloring $f$ of $G_t$ for which $f|_{\hat{\partial}(t)}=g$. Clearly, the problem is a Yes-instance if and only if there is a valid coloring at the root node.

Let
$\zeta=\max\{2^{\frac{(\alpha+w)(\alpha+w+1)}{2}}(\alpha+w)^{\alpha+w-1} (\alpha+w+1),(w+2\alpha-1)^\alpha\}.$

For each node $t\in V(T)$, let $Q[t]$ be the set of all valid colorings at $t$.
We will recursively construct a subset $Q^*[t]\subseteq Q[t]$  of size at most $\zeta$ satisfying that 
\begin{itemize}
    \item[($\ast$)] for every proper $L$-coloring $h$ on $G-V(G_t)$, if there is a valid coloring $g\in Q[t]$ compatible with $h$, then there is a valid coloring $g^*\in Q^*[t]$ compatible with $h$.
\end{itemize}
We describe how to construct $Q^*[t]$ depending on whether $t$ is a non-root leaf or not.
Let $t_p$ be the parent of $t$ when $t$ is not the root.

\medskip
\textbf{Case 1. $t$ is a non-root leaf.}

In this case, $G_t=G[X_t]$. 
Note that the degree of a vertex $v$ of $X_t$ in $G$ is at most $(\abs{X_t}-1)+\abs{X_{t_p}}+w\le w+2\alpha-1$, because there are at most $w$ edges that are crossing $X_{t_p}$.
Thus, $|L(v)|\le w+2\alpha-1$ for every $v\in X_t$ and this implies that $\abs{Q[t]}\le (w+2\alpha-1)^\alpha$. We take $Q^*[t]:=Q[t]$. This can be computed in time $\mathcal{O}(\zeta)$.

\medskip
\textbf{Case 2. $t$ is not a non-root leaf.}

We classify the children of $t$ into two types. Let $A_1$ be the set of all children $p$ of $t$ such that $V(\partial(p))\setminus \hat{\partial}(p)\subseteq X_t$, and let $A_2$ be the set of all other children of $t$. 
Note that $A_2$ is exactly the set of children $p$ of $t$ where $G_p$ is incident with some edge that crosses $X_t$.
Thus, $\abs{A_2}\le 2w$ because $(T, \mathcal{X})$ has crossing number at most $w$.
We assume that $Q^*[x]$ is computed for every child $x$ of $t$.
Let $C[t]$ be the set of all proper $L$-colorings on $X_t$. Clearly, $\abs{C[t]}\le n^\alpha$.

\medskip
(Step 1.) We first find the set $C'[t]$ of all proper $L$-colorings $f$ such that 
for each $x\in A_1$, there exists $g_x\in Q^*[x]$ that is compatible with $f$. This can be checked by recursively choosing $x\in A_1$, and comparing each coloring in $C[t]$ with a coloring in $Q^*[x]$, and then remaining one that has a compatible coloring in $Q^*[x]$. 
For fixed $x\in A_1$, this runs in time $\mathcal{O}(\abs{Q^*[x]}\cdot 
n^\alpha\cdot (\alpha+w)^2)$, and therefore, the whole procedure runs in time $\mathcal{O}(\abs{A_1}\cdot \abs{Q^*[x]}\cdot n^\alpha\cdot (\alpha+w)^2)=\mathcal{O}(\zeta \cdot n^{\alpha+1}\cdot (\alpha+w)^2 )$, because each $Q^*[x]$ has the size at most $\zeta$ and $\abs{A_1}\le n$.

\medskip
(Step 2.) Next, we compute the set $I[t]$ of all tuples $U$ in $\prod_{x\in A_2}Q^*[x]$ such that for all distinct $x,y\in A_2$, $U(x)$ and $U(y)$ are compatible, where $U(x)$ denotes the coordinate of $U$ that comes from $Q^*[x]$. 
Since $\abs{A_2}\le 2w$, we have $\abs{\prod_{x\in A_2}Q^*[x]}\le\zeta^{2w}$.
The set $I[t]$ can be computed in time $\mathcal{O}(\zeta^{2w}\cdot w^2\cdot  (\alpha+w)^2)$.

\medskip
(Step 3.) Lastly, we construct $Q'[t]$ from $I[t]$ and $C'[t]$ as follows.
For every $U\in I[t]$ and every $g\in C'[t]$ where $U(x)$ and $g$ are compatible for all $x\in A_2$, we obtain a new function $g'$ on $\hat{\partial}(t)$ such that 
\begin{itemize}
\item $g'(v)=(U(x))(v)$ if $v\in \hat{\partial}(x)$ for some $x\in A_2$, and \item $g'(v)=g(v)$ if $v\in X_t$,
\end{itemize}
and add it to $Q'[t]$.
This can be done in time \[\mathcal{O}(\abs{I[t]}\cdot \abs{C'[t]}\cdot (\alpha(\alpha+w))^{2w})=\mathcal{O}(\zeta^{2w}\cdot n^{\alpha}\cdot (\alpha(\alpha+w))^{2w}).\]
This stores valid colorings at $t$ and the size of $Q'[t]$ is at most $\abs{I[t]}\times \abs{C'[t]}$.
Using Lemma~\ref{lem:compgrph}, we find a subset $Q^*[t]$ of $Q'[t]$ of size at most $\zeta$.
    This can be computed in time \[\mathcal{O}(\abs{Q'[t]}\cdot 2^{(\alpha+w)^2}\cdot (\alpha+w)^{\alpha+w+2})=\mathcal{O}(\zeta^{2w}\cdot n^{\alpha}\cdot 2^{(\alpha+w)^2}\cdot (\alpha+w)^{\alpha+w+2}).\]
The total running time for this case is 
\[ \mathcal{O}(\zeta^{2w}\cdot n^{\alpha+1}\cdot 2^{(\alpha+w)^2}\cdot (\alpha+w)^{\alpha+4w+2}).  \]

For the correctness, we prove that $Q^*[t]$ satisfies the property ($\ast$).

\begin{claim}
The set $Q^*[t]$ satisfies the property ($\ast$).
\end{claim}
\begin{clproof}
Let $h$ be a proper $L$-coloring on $G-V(G_t)$, and suppose that there is a valid coloring $g\in Q[t]$ compatible with $h$. We need to prove that there is a valid coloring $g^*\in Q^*[t]$ that is compatible with $h$.
Since $g$ is a valid coloring at $t$, by definition, there is a proper $L$-coloring $f$ of $G_t$ for which $f|_{\hat{\partial}(t)}=g$.
By construction, the coloring $f|_{X_t}$ of $X_t$ is contained in $C[t]$.

Observe that  for every child $x$ of $t$, $f|_{V(G_x)}$ is a proper $L$-coloring of $G_x$, and thus, $f|_{\hat{\partial}(x)}\in Q[x]$.

Let $x\in A_1$.
As $f|_{\hat{\partial}(x)}$ is compatible with $f|_{V(G)\setminus V(G_x)}$, there exists $f_x^*\in Q^*[x]$ that is compatible with $f|_{V(G)\setminus V(G_x)}$.
Note that  $f^*_x$ is compatible with $f|_{X_t}$. This shows that $f|_{X_t}$ remains after considering all children in $A_1$ in Step 1, and $f|_{X_t}\in C'[t]$.    

If $A_2$ is empty, then this $f|_{X_t}$ is contained in $Q'[t]$. By Lemma~\ref{lem:compgrph}, there is a valid coloring $g^*\in Q^*[t]$ that is compatible with $h$. Thus, we may assume that $A_2$ is not empty.

Let $A_2=\{x_1, \ldots, x_d\}$.
We construct a sequence of proper $L$-colorings $f_0, f_1, \ldots, f_d$ of $G_t$ and colorings $g_1^*\in Q^*[x_1], \ldots, g_d^*\in Q^*[x_d]$ such that 
\begin{itemize}
    \item $f_0=f|_{V(G_t)}$,
    \item for each $i\in [d]$, $f_i=g_i^*\cup f_{i-1}|_{V(G_t)\setminus V(G_{x_i})}$.
\end{itemize}
Let $i\in [d]$ and suppose that $f_{i-1}$ has been constructed.  Note that $f_{i-1}|_{V(G_{x_i})}$ is compatible with $f|_{V(G)\setminus V(G_t)}\cup f_{i-1}|_{V(G_t)\setminus V(G_{x_i})}$. Thus, there exists $g_i^*\in Q^*[x_i]$ such that $g_i^*$ is compatible with $f|_{V(G)\setminus V(G_t)}\cup f_{i-1}|_{V(G_t)\setminus V(G_{x_i})}$.

Observe that  $\{g^*_i:i\in [d]\}$ are pairwise compatible. Thus, the tuple $U$ where $U(x_i)=g^*_i$ is added to $I[t]$ in Step 2. Note that $U(x)$ and $f|_{X_t}$ are compatible for all $x\in A_2$. So, the function $g'$ described in Step 3 is added to $Q'[t]$.

 Note that when we obtain $f^*_x$ from  $f|_{\hat{\partial}(x)}$, we have the property that $f^*_x$ is compatible with $h$. Thus, $g'$ is compatible with $h$. So, by the construction of $Q^*[t]$, there is $g''\in Q^*[t]$ that is compatible with $h$.
\end{clproof}

As $|V(T)|=\mathcal{O}(n)$, the algorithm runs in time 
\[ \mathcal{O}(\zeta^{2w}\cdot n^{\alpha+2}\cdot 2^{(\alpha+w)^2}\cdot (\alpha+w)^{\alpha+4w+2})=2^{\mathcal{O}((\alpha^2+k)^3)}n^{\alpha+2}.\qedhere\]
\end{proof}

\begin{corollary}\label{cor:precoloring}
For a fixed positive integer $\alpha$, the \textsc{Precoloring Extension} problem is fixed parameter tractable parameterized by $\alpha$-edge-crossing width.
\end{corollary}
\begin{proof}
Let $(G, S, q, c_S)$ be an instance.
We assign a list $L(v)$ for every vertex $v$ of $G$ as follows.
\begin{itemize}
    \item For every $v\notin S$, let $L(v)=[q]\setminus \{c_S(u):u\in N_G(v)\cap S\}$.
    \item For every $v\in S$, let $L(v)=\{c_S(v)\}$.
\end{itemize}
This spends $\mathcal{O}(n)$ time.
Using Theorem~\ref{thm:lcalpha}, in $2^{\mathcal{O}((\alpha^2+k)^3)}n^{\alpha+2}$ time, one can either get a proper $L$-coloring $c$ of $G$ or report that there is no proper $L$-coloring of $G$.
By the construction of lists, for every $x\in S$, $x$ must be colored by $c_S(x)$. 
So, the coloring $c$ is also a solution of the \textsc{Precoloring Extension} problem.
If there is no proper $L$-coloring of $G$, then this implies that there is no coloring $c$ with $c(v)=c_S(v)$ for all $v\in S$.
This proves the statement.
\end{proof}

\section{Relationships between width parameters}\label{sec:relationship}

 In this section, we compare width parameters as presented in Figure~\ref{fig:parameters}. 

For two functions $\phi$ and $\psi$ defined on graphs, we write $\psi\WR\phi$ if there is a function $f$ such that $\phi(G)\le f(\psi(G))$ for all graphs $G$. Otherwise, we write $\phi\not\WR\psi$.
We say that two functions $\phi$ and $\psi$ are \emph{incomparable} when $\phi\not\WR\psi$ and $\psi\not\WR\phi$.

\subsection{Edge-crossing width and 
$\alpha$-edge-crossing width}\label{subsec:ecrwandalpha}
We start with proving basic properties of edge-crossing width and $\alpha$-edge-crossing width.

\begin{lemma}\label{lem1st}
Let $\alpha$ and $\beta$ be positive integers.
\begin{enumerate}
    \item If $\alpha<\beta$, then $\ecrw_\beta(G)\le \ecrw_\alpha(G)$ for every graph    $G$. Thus, $\ecrw_\alpha\WR\ecrw_\beta$.
    \item For every graph $G$, $\ecrw(G)\le \max(\alpha,\ecrw_\alpha(G))$. Thus, $\ecrw_\alpha\WR\ecrw$.
\end{enumerate}
\end{lemma}
\begin{proof}
Let $G$ be a graph.

(1) Assume that $\alpha<\beta$. Let $(T,\mathcal{X})$ be a tree-cut decomposition of $G$ with thickness at most $\alpha$ and crossing number $\ecrw_\alpha(G)$. Since $\alpha<\beta$, $(T,\mathcal{X})$ has thickness at most $\beta$ and crossing number $\ecrw_\alpha(G)$.  Thus, the $\beta$-edge-crossing width of $G$ is at most $\ecrw_\alpha(G)$.

(2) Let $(T,\mathcal{X})$ be a tree-cut decomposition of $G$ with thickness at most $\alpha$ and crossing number $\ecrw_\alpha(G)$. Then it has edge-crossing width at most $\max (\alpha,\ecrw_\alpha(G))$.
It follows that $\ecrw(G)\le \max(\alpha,\ecrw_\alpha(G))$.
\end{proof}

Next, we prove that $\ecrw_\beta\not\WR\ecrw_\alpha$ if $\alpha<\beta$.
For all positive integers $k$ and $n$, we construct a graph $G^n_k$ as follows.
Let $A:=\{a_i: i\in [n]\}$, $B={A\choose 2}$, and $B_k=\{(W,\ell):W\in B\text{ and } \ell\in[k]\}$.
Let $G_k^n$ be the graph such that
\begin{itemize}
    \item $V(G_k^n)=A\cup B_k$, and
    \item for $a\in A$ and $(W, \ell)\in B_k$, $a$ is adjacent to $(W, \ell)$ in $E(G_k^n)$ if and only if $a\in W$.
\end{itemize}

\begin{lemma}\label{lem:loweralpha}
 The graph $G^{\alpha+1}_{3k+3\alpha-1}$ has $\alpha$-edge-crossing width at least $k+1$.
\end{lemma}
\begin{proof}
Suppose for contradiction that $\ecrw_\alpha(G^{\alpha+1}_{3k+3\alpha-1})\le k$.
    Let $(A, B_{3k+3\alpha-1})$ be the bipartition of $G^{\alpha+1}_{3k+3\alpha-1}$ given by the definition. 
 Let $\mathcal{T}=(T, \{X_t\}_{t\in V(T)})$ be a tree-cut decomposition of $G^{\alpha+1}_{3k+3\alpha-1}$ with thickness at most $\alpha$ and crossing number $\ecrw_\alpha(G^{\alpha+1}_{3k+3\alpha-1})$.
 Since each bag has at most $\alpha$ vertices of $G^{\alpha+1}_{3k+3\alpha-1}$, there are two distinct nodes $p$ and $q$ of $T$ such that each of $X_p$ and $X_q$ contains a vertex of $A$. Let $a_1\in A\cap X_p$ and $a_2\in A\cap X_q$.
 
 Let $\mathcal{C}$ be the set of connected components of $T-\{p,q\}$. Note that there is at most one connected component of $\mathcal{C}$ that has neighbors of both $p$ and $q$. Let $T^*$ be this component if one exists. 
 For each $y\in \{p,q\}$, let $T_y$ be the union of the connected components in $\mathcal{C}$ that has a neighbor of $y$ and does not have a neighbor of the vertex of $\{p,q\}\setminus \{y\}$.
 
 Let $Z$ be the set of vertices in $B_{3k+3\alpha-1}$ whose neighborhoods are exactly $\{a_1, a_2\}$.
 If $\bigcup_{t\in V(T_p)} X_t$ contains $k+1$ vertices of $Z$, then $\cross_\mathcal{T}(p)\ge k+1$. 
 If $\bigcup_{t\in V(T_q)} X_t$ contains $k+1$ vertices of $Z$, then $\cross_\mathcal{T}(q)\ge k+1$. Therefore, we may assume that each of $\bigcup_{t\in V(T_p)} X_t$  and $\bigcup_{t\in V(T_q)} X_t$ contains at most  $k$ vertices of $Z$.
 Also, $X_p\cup X_q$ may contain at most $2(\alpha-1)$ vertices of $Z$.
 Since $(3k+3\alpha-1)-2k-2(\alpha-1)=k+\alpha+1$,
 $T^*$ exists and $\bigcup_{t\in V(T^*)}X_t$ contains at least $k+\alpha+1$ vertices of $Z$. 
 
 Let $p^*$ be the neighbor of $p$ contained in $T^*$.
 Since $X_{p^*}$ contains at most $\alpha$ vertices, there are at least $k+1$ vertices of $Z$ contained in $\bigcup_{t\in V(T^*)\setminus \{p^*\}} X_t$.
 The edges between $a_1$ and these vertices cross $X_{p^*}$, and therefore, $\cross_\mathcal{T}(p^*)\ge k+1$. This implies the result.
\end{proof}

\begin{lemma}
Let $\alpha$ and $\beta$ be positive integers with $\alpha<\beta$. Then $\ecrw_\beta\not\WR\ecrw_\alpha$.
\end{lemma}
\begin{proof}
Let $\mathcal{B}=\{G^{\alpha+1}_n:n\in \mathbb{N}\}$. 
By Lemma~\ref{lem:loweralpha}, $\mathcal{B}$ has unbounded $\alpha$-edge-crossing width. 
We claim that for every $k$, $G^{\alpha+1}_k$ has $\beta$-edge-crossing width $0$.
Let $(A, B_k)$ be the bipartition of $G^{\alpha+1}_k$ given by the definition.
Let $T$ be a star with center $t$ and leaves $t_1, \ldots, t_{k{\beta\choose 2}}$. 
Let $X_t=A$ and each $X_{t_i}$ consists of a vertex of $B_k$. It is clear that $(T, \{X_v\}_{v\in V(T)})$ is a tree-cut decomposition of thickness at most $\alpha+1\le \beta$ and crossing number $0$.
\end{proof}

\subsection{$\alpha$-edge-crossing width, tree-cut width, and slim tree-cut width }\label{subsec:ecrwandtcw}

In this subsection, we show that $\stcw\WR \ecrw_\alpha$ and $\ecrw_\alpha\not\WR \stcw$. We also show that $\ecrw_\alpha\not\WR \tcw$  and $ \tcw\not\WR \ecrw_\alpha$.

For positive integers $k$ and $n$, let $S_{k,n}$ be the graph obtained from $K_{1,n}$ by replacing each edge with $k$ internally vertex-disjoint paths of length $2$. See Figure~\ref{fig:S3nDecomposition} for an illustration.
Ganian and Korchemna~\cite[Lemma 3]{Ganian2022} proved that the slim tree-cut width of $S_{2,n^2}$ is at least $n$. 

\begin{lemma}[Ganian and Korchemna~\cite{Ganian2022}]\label{lem:unbddstcw}
    The set $\{S_{2,n}:n\in \mathbb{N}\}$ has unbounded slim tree-cut width.
\end{lemma}

We further show the following.

\begin{lemma}\label{lem31}
The following statements hold.
\begin{enumerate}[(1)]
    \item $\{S_{3,n}:n\in \mathbb{N}\}$ has 1-edge-crossing width at most $2$.
    \item $\{S_{3,n}:n\in \mathbb{N}\}$ has unbounded tree-cut width.
\end{enumerate}
\end{lemma}
\begin{proof}
Let 
\begin{itemize}
    \item $V(S_{3,n})=\{u_i:i\in \{0\}\cup [n]\}\cup\{v_{i,j}:i\in [n], j\in [3]\}$ and
    \item $E(S_{3,n})=\{u_0v_{i,j}, u_{i}v_{i,j}:i\in [n], j\in [3]\}$.
\end{itemize} 

(1) 
We construct a tree-cut decomposition $\mathcal{T}=(T,\mathcal{X})$ of $S_{3,n}$ which has thickness $1$ and crossing number at most $2$.
Let $T$ be the graph obtained from $K_{1,n}$ by replacing each edge with a path of length $4$. 
Let $c$ be the vertex of degree $n$ in $T$, and 
for $i\in[n]$, let $\{t_{i, j}:i\in[n], j\in [4]\}$ be a set of vertices of $T$ where $ct_{i, 1}t_{i, 2}t_{i, 3}t_{i, 4}$ is a subdivided path.

\begin{figure}[t]
    \centering
    \includegraphics[scale=0.9]{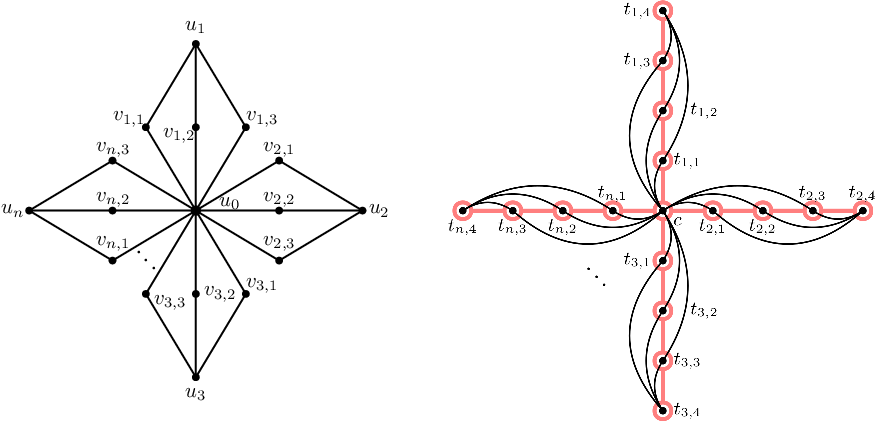}
    \caption{The graph $S_{3,n}$ and its decomposition.}
    \label{fig:S3nDecomposition}
\end{figure}
    
Let $X_c=\{u_0\}$. For each $i\in[n]$ and $j\in[3]$, let $X_{t_{i, j}}=\{v_{i, j}\}$. For each $i\in[n]$, let $X_{t_{i, 4}}=\{u_i\}$. Let $\mathcal{X}=\{X_v\}_{v\in V(T)}$.
See Figure~\ref{fig:S3nDecomposition}.
Note that the size of each bag is exactly $1$. Observe that $\cross_\mathcal{T}(t_{i,j})=2$ for all $i\in[n]$ and $j\in [3]$, and $\cross_\mathcal{T}(t_{i, 4})=0$ for all $i\in[n]$, and $\cross_\mathcal{T}(c)=0$. Hence $\ecrw_1(S_{3,n})\le 2$.

(2) 
Wollan~\cite[Theorem 15]{Wollan2015} proved that if $n$ is large, then $S_{3,n}$ contains a large wall as a weak immersion, and therefore, $\{S_{3,n}:n\in \mathbb{N}\}$ has unbounded tree-cut width. 
\end{proof}

\begin{lemma}
For every positive integer $\alpha$,
$\ecrw_\alpha\not\WR\stcw$ and $\ecrw_\alpha\not\WR\tcw$.
\end{lemma}
\begin{proof}
Note that $S_{2,n}$ is isomorphic to an induced subgraph of $S_{3,n}$. So, by (1) of Lemma~\ref{lem31}, $S_{2,n}$ has $1$-edge crossing width at most $2$. On the other hand, Lemma~\ref{lem:unbddstcw} shows that $\{S_{2,n}:n\in \mathbb{N}\}$ has unbounded slim tree-cut width. This shows that $\ecrw_1\not\WR\stcw$. Since $\ecrw_1\WR \ecrw_\alpha$, we have $\ecrw_\alpha\not\WR\stcw$.

By (1) and (2) of Lemma~\ref{lem31}, $\{S_{3,n}:n\in \mathbb{N}\}$ has $1$-edge crossing width at most $2$, but unbounded tree-cut width. Therefore,  
$\ecrw_1\not\WR\tcw$.
Since $\ecrw_1\WR \ecrw_\alpha$, we have $\ecrw_\alpha\not\WR\tcw$.
\end{proof}

We now show that $\tcw\not\WR\ecrw_\alpha$. This implies that $\alpha$-edge crossing width and tree-cut width are incomparable.

\begin{lemma}\label{lem32}
For every positive integer $\alpha$, $\tcw\not\WR\ecrw_\alpha$.
\end{lemma}
\begin{proof}
We recall the graphs $G^n_k$ constructed in Subsection~\ref{subsec:ecrwandalpha}. 
By Lemma~\ref{lem:loweralpha}, $\{G^{\alpha+1}_k:k\in \mathbb{N}\}$ has unbounded $\alpha$-edge-crossing width. 

We claim that for every $k$, $G^{\alpha+1}_k$ has tree-cut width at most $\alpha+1$.
Let $(A, B_k)$ be the bipartition of $G^{\alpha+1}_k$ given by the definition.
Let $T$ be a star with center $t$ and leaves $t_1, \ldots, t_{k{\alpha+1\choose 2}}$. 
Let $X_t=A$ and each $X_{t_i}$ consists of a vertex of $B_k$.
Let $\mathcal{T}=(T, \{X_v\}_{v\in V(T)})$.
Note that the $3$-center of $H_t$ has only vertices of $A$. Thus, it has at most $\alpha+1$ vertices. Also, for every edge $e$ of $T$, $\adh_{\mathcal{T}}(e)\le 2$.  
So, $\mathcal{T}$ is a tree-cut decomposition of tree-cut width at most $\alpha+1$.
\end{proof}

Lastly, we show $\stcw\WR\ecrw_\alpha$. For this, we use another parameter called \emph{super edge-cut width} introduced by Ganian and Korchemna~\cite{Ganian2022}. It was shown that it is equivalent to slim tree-cut width.
The \emph{super edge-cut width} of a graph $G$, denoted by $\secw(G)$, is defined as the minimum edge-cut width of $(H, T)$ over all supergraphs $H$ of $G$ and maximal spanning forests $T$ of $H$. 
\begin{theorem}[Ganian and Korchemna~\cite{Ganian2022}]\label{thm:superecw}
    $\secw\WR \stcw$ and $\stcw\WR \secw$.
\end{theorem}
We show that $\secw\WR \ecrw_\alpha$ for every $\alpha$.

\begin{lemma}\label{lem:secwecrw}
    For every positive integer $\alpha$, $\secw\WR \ecrw_\alpha$.
\end{lemma}
\begin{proof}
    It is sufficient to show that $\secw\WR\ecrw_1$. Let $G$ be a graph with $\secw(G)=k$. Then there exist a supergraph $H$ of $G$ and a maximal spanning forest $F$ of $H$ for which the edge-cut width of $(H, F)$ is $k$. 
    Since the edge-cut width of $(H,F)$ is $k$, we have $|E^{H,F}_{loc}(v)|\le k-1$ for every vertex $v$ of $H$. 

    We construct a tree-cut decomposition of $G$ by extending $F$ to a tree.
    Let $T$ be a tree on $V(F)$ containing $F$ as a subgraph. Observe that every edge $e\in E(T)\setminus E(F)$ connects two distinct components of $F$.
    For every $t\in V(T)$, let $X_t=\{t\}\cap V(G)$.
    Then  $\mathcal{T}=(T, \{X_t\}_{t\in V(T)})$ is a tree-cut decomposition of $G$ of thickness $1$.
    
    We claim that 
$\cross_\mathcal{T}(t)\le k-1$ for every $t\in V(T)$.    
    Let $t\in V(T)$. 
    Suppose $e=uv$ is an edge of $G$ crossing $X_t$. 
    As $u$ and $v$ are contained in a component of $G$, they are also contained in a component of $F$. Therefore, the unique path from $u$ to $v$ in $T$ does not contain an edge of $E(T)\setminus E(F)$. Thus, $e\in E^{H,F}_{loc}(v)$. 
    This implies that 
    $\cross_\mathcal{T}(t)\le \abs{E^{H,F}_{loc}(v)}\le k-1$. 
    
    We deduce that $\ecrw_1(G)\le k-1$.
\end{proof}

\subsection{Edge-crossing width and tree-partition-width}\label{subsec:ecrwtpw}
In this subsection, we show that edge-crossing width and tree-partition-width are equivalent.

\begin{lemma}
For every graph $G$, $\ecrw(G)\le\tpw(G)$. Thus, $\tpw\WR\ecrw$.
\end{lemma}
\begin{proof}
Let $(T,\mathcal{X})$ be a tree-partition of $G$ whose thickness is $\tpw(G)$. By definition, $(T,\mathcal{X})$ is a tree-cut decomposition, and for every node $t$ of $T$, we have $\cross_{\mathcal{T}}(t)=0$. Thus, $\ecrw(G)\le\tpw(G)$.
\end{proof}

To show $\ecrw\WR\tpw$, we use a characterization of graphs of bounded tree-partition-width given by Ding and Oporowski~\cite{tpw1996}. 
Let $n$ be a positive integer. An $n$-grid is a graph whose vertex set is $\{(i,j):i,j\in [n]\}$ and edge set is $\{\{(i,j),(i',j')\}:i,j\in[n], |i-i'|+|j-j'|=1\}$.
An \emph{$n$-fan} is a graph $F_n$ obtained from the path graph on $n$ vertices by adding a vertex adjacent to all vertices of the path.
A \emph{thickened $n$-star} is $S_{n,n}$ constructed in Subsection~\ref{subsec:ecrwandtcw}.
A \emph{thickened $n$-path} is a graph obtained from the path graph on $n$ vertices by replacing each edge by $n$ internally vertex-disjoint paths of length two.
An \emph{$n$-wall} is a graph constructed from $n$-grid by
\begin{itemize}
    \item deleting all edges of the form $\{(2i,2j-1),(2i,2j)\}$ for all positive integers $i$ and $j$ such that $2\le 2i\le n$ and $1\le 2j\le n$, 
    \item deleting all edges form $\{(2i-1,2j),(2i-1,2j+1)\}$ for all integers $i$ and $j$ such that $1\le 2i-1\le n$ and $2\le 2j+1\le n$, and then 
    \item deleting vertices of degree $1$.
\end{itemize} 
See Figure~\ref{fig:FourFigures} for examples.

\begin{figure}
    \centering
    \includegraphics[scale=0.7]{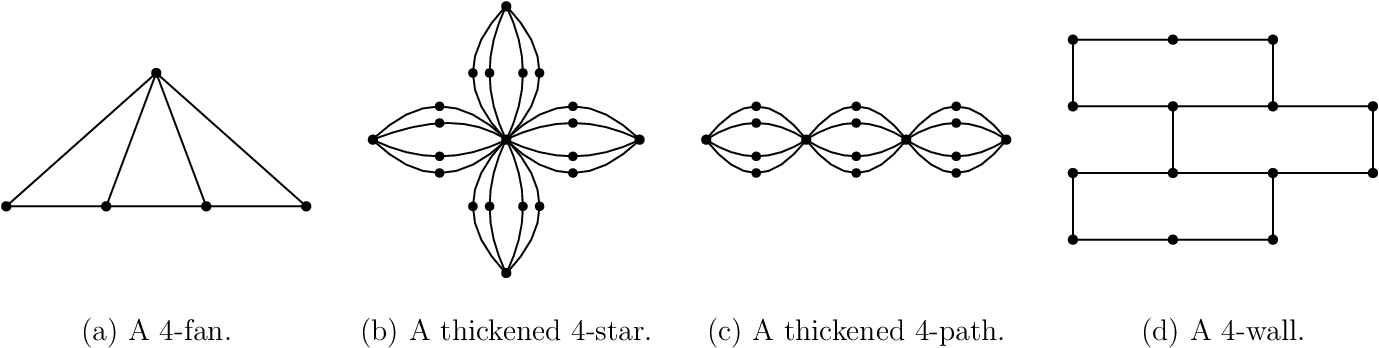}
    \caption{An example of an $n$-fan, a thickened $n$-star, a thickened $n$-path, and an $n$-wall, when $n=4$.}
    \label{fig:FourFigures}
\end{figure}

\begin{theorem}[Ding and Oporowski~\cite{tpw1996}]\label{tpw-iff}
Let $\gamma$ be a function such that for every graph $G$, $\gamma(G)$ is the maximum integer $r$ for which $G$ contains a subdivision of an $r$-fan, a thickened $r$-star, a thickened $r$-path, or an $r$-wall.
    Then $\gamma\WR \tpw$ and $\tpw\WR \gamma$. 
\end{theorem}

We first show that if $H$ is a subdivision of $G$, then $\ecrw(G)\le 2\ecrw(H)$. Then we verify that each of an $r$-fan, a thickened $r$-star, a thickened $r$-path, and an $r$-wall has large edge-crossing width if $r$ is large. This will imply $\ecrw\WR\tpw$ by Theorem~\ref{tpw-iff}.

\begin{lemma}\label{lem:subdivi}
    Let $G$ and $H$ be graphs.
    If $H$ is a subdivision of $G$, then $\ecrw(G)\le 2\ecrw(H)$.
\end{lemma}
\begin{proof}
Let $\mathcal{T}=(T,\{X_t\}_{t\in V(T)})$ be a tree-cut decomposition of $H$ of edge-crossing width $\ecrw(H)$.
For every edge $uv$ of $G$, if $uw_1\cdots w_mv$ is the subdivided path in $H$, then 
let $I(uv)=\{w_1, \ldots, w_m\}$. Let $I=\bigcup_{uv\in E(G)}I(uv)$, and let $Y_t=X_t\setminus I$ for all $t\in V(T)$.

We claim that $\mathcal{T}'=(T, \{Y_t\}_{t\in V(T)})$ is a tree-cut decomposition of $G$ of edge-crossing width at most $2\ecrw(H)$. Clearly, the thickness of $\mathcal{T'}$ is at most the thickness of $\mathcal{T}$. Thus, it is sufficient to show that the crossing number of $\mathcal{T'}$ is at most $2\ecrw(H)$.

Let $t$ be a node of $T$. If $V(T)=\{t\}$, then there is nothing to prove. We may assume that $T$ has at least two nodes. Let $T_1, \ldots, T_m$ be the connected components of $T-t$, and for each $i\in [m]$, let $V_i=\bigcup_{p\in V(T_i)}Y_p$. 

Let $F$ be the set of edges crossing $Y_t$ in $\mathcal{T}'$. Let $uv\in F$, and let $uw_1 \cdots w_m v$ be the subdivided path in $G$. As $uv$ crosses $Y_t$, $u$ and $v$ are contained in distinct sets of $Y_1, \ldots, Y_m$. Observe that if there is no edge in the path $uw_1\cdots w_mv$ crossing $X_t$ in $\mathcal{T}$, then 
one of $w_1, \ldots, w_m$ must be contained in $X_t$.
Therefore, the latter case can happen for at most $|X_t|$ many edges of $F$. This implies that $\cross_{\mathcal{T}'}(t)\le \cross_\mathcal{T}(t)+|X_t|\le 2\ecrw(H)$, as required.
\end{proof}

\begin{lemma}\label{lem:thickenpath}
Let $k$ be a positive integer. For a thickened $4k$-path $H$, $\ecrw(H)\ge k$.
\end{lemma}
\begin{proof}
Let 
\begin{itemize}
    \item $V(H)=\{u_i:i\in [4k]\}\cup\{v_{i,j}:i\in [4k-1], j\in[4k]\}$ and
    \item $E(H)=\{u_iv_{i,j}, u_{i+1}v_{i,j}:i\in [4k-1], j\in[4k]\}$.
\end{itemize} 

Suppose for contradiction that 
$H$ admits a tree-cut decomposition $\mathcal{T}=(T,\{X_t\}_{t\in V(T)})$ of edge-crossing width less than $k$. 
As each bag has size less than $k$, one bag does not contain all vertices in $\{u_i:i\in [4k]\}$. This implies there exists $a\in [4k-1]$ such that $u_a$ and $u_{a+1}$ are contained in distinct bags of $\mathcal{T}$. Let $X_{p}$ and $X_{q}$ be the bags contain $u_a$ and $u_{a+1}$, respectively.

Let $pp^*$ be the edge of $T$ on the path from $p$ to $q$ in $T$. Let $T_p$ and $T_{p^*}$ be the components of $T-pp^*$ containing $p$ and $p^*$, respectively. Let $V_p=\bigcup_{t\in V(T_p)}X_t$ and $V_{p^*}=\bigcup_{t\in V(T_{p^*})}X_t$.
Note that either $V_p$ or $V_{p^*}$ contains at least 
$2k$ vertices in $\{v_{a,j}:j\in[4k]\}$. 

If $V_p$ contains at least 
$2k$ vertices in $\{v_{a,j}:j\in[4k]\}$, then $\cross_{\mathcal{T}}(p)\ge k$, because $|X_p|<k$.
Similarly, 
if $V_{p^*}$ contains at least 
$2k$ vertices in $\{v_{a,j}:j\in[4k]\}$, then $\cross_{\mathcal{T}}(p^*)\ge k$, because $|X_{p^*}|<k$.
These contradict the assumption that $\mathcal{T}$ has edge-crossing width less than $k$.
\end{proof}

\begin{lemma}\label{lem:thickenstar}
Let $k$ be a positive integer. For a thickened $4k$-star $H$, $\ecrw(H)\ge k$.
\end{lemma}
\begin{proof}
Let 
\begin{itemize}
    \item $V(H)=\{c\}\cup \{u_i:i\in [4k]\}\cup\{v_{i,j}:i\in [4k], j\in[4k]\}$ and
    \item $E(H)=\{cv_{i,j}, u_{i}v_{i,j}:i\in [4k], j\in[4k]\}$.
\end{itemize} 

Suppose for contradiction that 
$H$ admits a tree-cut decomposition $\mathcal{T}=(T,\{X_v\}_{v\in V(T)})$ of edge-crossing width less than $k$. 
As each bag has size less than $k$, the bag containing $c$ does not contain all vertices in $\{u_i:i\in [4k]\}$. Let $a\in [4k]$ such that $c$ and $u_a$ are contained in distinct bags of $\mathcal{T}$. Let $X_{p}$ and $X_{q}$ be the bags containing $c$ and $u_{a}$, respectively.

Let $pp^*$ be the edge of $T$ on the path from $p$ to $q$ in $T$. Let $T_p$ and $T_{p^*}$ be the components of $T-pp^*$ containing $p$ and $p^*$, respectively. Let $V_p=\bigcup_{t\in V(T_p)}X_t$ and $V_{p^*}=\bigcup_{t\in V(T_{p^*})}X_t$.
Note that either $V_p$ or $V_{p^*}$ contains at least 
$2k$ vertices in $\{v_{a,j}:j\in[4k]\}$. 

If $V_p$ contains at least 
$2k$ vertices in $\{v_{a,j}:j\in[4k]\}$, then $\cross_{\mathcal{T}}(p)\ge k$, because $|X_p|<k$.
Similarly, 
if $V_{p^*}$ contains at least 
$2k$ vertices in $\{v_{a,j}:j\in[4k]\}$, then $\cross_{\mathcal{T}}(p^*)\ge k$, because $|X_{p^*}|<k$.
These contradict the assumption that $\mathcal{T}$ has edge-crossing width less than $k$.
\end{proof}

\begin{lemma}\label{lem:fanecrw}
    Let $k\ge 3$ be an integer. For a $(6k^2-6k-1)$-fan $F_{6k^2-6k-1}$, $\ecrw(F_{6k^2-6k-1})\ge k$. 
\end{lemma}
\begin{proof}
Let $G=F_{6k^2-6k-1}$ with the vertex $u$ of degree $6k^2-6k-1$.
Suppose not.
Let $\mathcal{T}=(T,\{X_v\}_{v\in V(T)})$ be a tree-cut decomposition of $G$ with edge-crossing width less than $k$. For each node $t$ of $T$, let $\mathcal{C}_t$ be the set of all components of $T-t$. Since $k\ge 3$, we have $6k^2-6k-1\ge 2k$. So, $G$ has more than $2k$ vertices. 

\begin{claim}
There are a node $q$ of $T$ and a non-empty subset $S$ of $\mathcal{C}_q$ such that
\[
        \frac{1}{4}|V(G)|<\sum_{C\in S}\left(\sum_{v\in V(C)}|X_v|\right)\le \frac{3}{4}|V(G)|.
        \]
\end{claim}
\begin{clproof}
For each edge $ab\in E(T)$, let $T_a$ and $T_b$ be the components of $T-ab$ containing $a$ and $b$, respectively. If $\sum_{t\in V(T_a)}\abs{X_t}\le\sum_{t\in V(T_b)}\abs{X_t}$, then we direct the edge $ab$ from $a$ to $b$.
Since each leaf bag $X_t$ contains at most $k$ vertices and $G$ has more than $2k$ vertices, the edge of $T$ incident with~$t$ is directed from $t$ to its neighbor.

So, there is a node $q$ such that all edges incident with $q$ are directed to $q$. We claim that $q$ satisfies the required property.
Suppose not. Let $S$ be an inclusion-wise maximal subset of $\mathcal{C}_q$ where $\sum_{C\in S}\left(\sum_{t\in V(C)}|X_t|\right)\le \frac{1}{4}|V(G)|$.
Since $\abs{X_q}\le k<\frac{1}{2}|V(G)|$, we have $\mathcal{C}_q\setminus S\neq \emptyset$. Let $C^*\in \mathcal{C}_q\setminus S$. Because of the directions defined on edges of $T$, we have $\left(\sum_{t\in V(C^*)}|X_t|\right)\le \frac{1}{2}\abs{V(G)}$. Therefore, we have 
$\frac{1}{4}|V(G)|< \sum_{C\in S\cup \{C^*\}}\left(\sum_{t\in V(C)}|X_t|\right)\le \frac{1}{4}|V(G)|+\frac{1}{2}|V(G)|=\frac{3}{4}|V(G)|.$ 
\end{clproof}

        First assume that $X_q$ does not contain the vertex $u$. Then either  
        $\bigcup_{C\in S}\left(\bigcup_{t\in V(C)}X_t\right)$ contains $u$ or 
        $\bigcup_{C\in \mathcal{C}_q\setminus S}\left(\bigcup_{t\in V(C)}X_t\right)$ contains $u$. Assume that $\bigcup_{C\in S}\left(\bigcup_{t\in V(C)}X_t\right)$ contains $u$.
        By the choice of $q$ and $S$, the union of $\bigcup_{C\in \mathcal{C}_q\setminus S}\left(\bigcup_{t\in V(C)}X_t\right)$ and $X_q$ contains more than $\frac{6k^2-6k}{4}\ge 2k-1$ vertices of $F_{6k^2-6k-1}$. As $\abs{X_q}<k$, 
        $\bigcup_{C\in \mathcal{C}_q\setminus S}\left(\bigcup_{v\in V(C)}X_v\right)$ contains at least $k$ vertices of $F_{6k^2-6k-1}$.
        This implies that $\cross_\mathcal{T}(q)\ge k$, a contradiction. A similar argument shows that if $\bigcup_{C\in \mathcal{C}_q\setminus S}\left(\bigcup_{v\in V(C)}X_v\right)$ contains $u$, then $\cross_\mathcal{T}(q)\ge k$, which leads to a contradiction.

    Therefore, we may assume that $X_q$ contains $u$. Let $P$ be the component of the graph obtained from $F_{6k^2-6k-1}$ by removing all vertices in $X_q$ and all vertices incident with edges crossing $X_q$, such that $\abs{V(P)}$ is maximum. Since we remove at most $(2k-2)+(k-2)=3k-4$ vertices of the path $F_{6k^2-6k-1}-u$, there are at most $3k-3$ components and therefore, $P$ has at least  
    $\frac{6k^2-6k-1-(3k-4)}{3k-3}=2k-1$ vertices.

    Since we delete all vertices incident with crossing edges, $P$ is fully contained in the part corresponding to one of the components of $T-q$. Let $T^*$ be the component of $T-q$ where $V(P)$ is contained in $\bigcup_{v\in V(T^*)}X_v$.
    Let $q'$ be the node in $T^*$ adjacent to $q$.
    Since $X_{q'}$ has less than $k$ vertices, 
    $P$ has at least $k$ vertices that is contained in
    $\left(\bigcup_{v\in V(T^*)}X_v\right)\setminus X_{q'}$.
    This implies that $\cross_\mathcal{T}(q')\ge k$, a contradiction.

    We conclude that $F_{6k^2-6k-1}$ has edge-crossing width at least $k$.
\end{proof}

Now, we show that $\ecrw\WR\tpw$.
\begin{lemma}\label{lem:tpwecrw}
    $\ecrw\WR\tpw$.
\end{lemma}
\begin{proof}
Let $\gamma$ be the function defined in Theorem~\ref{tpw-iff}. 
As  $\gamma\WR \tpw$, there is a non-decreasing function $f$ such that for every graph $G$, $\tpw(G)\le f(\gamma(G))$. Let $g(x)=f( 24x^2+36x+11 )$. We claim that for every graph $G$, $\tpw(G)\le g(\ecrw(G))$. This will imply that $\ecrw\WR \tpw$.

Suppose that $G$ has edge-crossing width $k$. Then by Lemmas~\ref{lem:thickenpath}, \ref{lem:thickenstar}, \ref{lem:fanecrw} together with \ref{lem:subdivi}, $G$ has no subdivision isomorphic to $F_{24(k+1)^2-12(k+1)-1}=F_{24x^2+36x+11}$, a thickened $(8k+8)$-path, or a thickened $(8k+8)$-star.
Since an $n$-wall contains a minor isomorphic to an $n$-grid and an $n$-grid has tree-width $n$, an $n$-wall has tree-width at least $n$ and by Lemma~\ref{lem:twecrw}, it has edge-crossing width at least $\frac{n+1}{5}$. Furthermore, by Lemma~\ref{lem:subdivi}, a subdivision of an $n$-wall has edge-crossing width at least $\frac{n+1}{10}$. This implies that  if $G$ has a subdivision of a $(10k+9)$-wall, then $\ecrw(G)\ge k+1.$ Therefore, $G$ has no subdivision of a $(10k+9)$-wall.

Thus, $\gamma(G)\le \max(24k^2+36k+11, 8k+8, 10k+9)=24k^2+36k+11$. Since $f$ is non-decreasing, $\tpw(G)\le f(\gamma(G))\le f(24k^2+36k+11)=g(k)$.
\end{proof}

It is an open problem to find a simple and direct upper bound of tree-partition-width in terms of edge-crossing width.

\section{Conclusion}\label{sec:conclusion}

In this paper, we introduced a width parameter called $\alpha$-edge-crossing width, which lies between edge-cut width and tree-partition-width, and which is incomparable with tree-cut width. We showed that \textsc{List Coloring} and \textsc{Precoloring Extension} are FPT parameterized by $\alpha$-edge-crossing width. It would be interesting to find more problems that are W[1]-hard parameterized by tree-partition-width, but FPT by $\alpha$-edge-crossing width for any fixed $\alpha$. There are five more problems that are known to admit FPT algorithms parameterized by slim tree-cut width, but W[1]-hard parameterized by tree-cut width~\cite{Ganian2022}, and these problems are candidates for the next research.  

We remark that the \textsc{Edge-Disjoint Paths} problem is one of the W[1]-hard problems parameterized by tree-width that motivates to study width parameters based on edge cuts. 
    Fleszar, Mnich, and Spoerhase~\cite{FleszarMS2018} proved that \textsc{Edge-Disjoint Paths} is NP-hard on graphs admitting a vertex cover of size $3$ and graphs admitting a feedback vertex set of size $2$. This implies that for every $\alpha\ge 3$, this problem is NP-hard on graphs of $\alpha$-edge-crossing width $0$. 
    For $\alpha=1$, one can easily design a fixed parameter algorithm  for \textsc{Edge-Disjoint Paths} parameterized by $\alpha$-edge-crossing width, similar to the algorithm parameterized by edge-cut width in \cite{BrandCHGK2022}. For $\alpha=2$, we leave the problem of deciding whether \textsc{Edge-Disjoint Paths} is FPT parameterized by $\alpha$-edge-crossing width as an open problem. 

We also introduced edge-crossing width, that is equivalent to tree-partition-width. However, our proof is based on the characterization of graphs of bounded tree-partition-width due to Ding and Oporowski~\cite{tpw1996}, and finding an elementary upper bound of tree-partition-width in terms of edge-crossing width is an interesting problem. More specifically, we ask whether there is a constant $c$ such that for every graph $G$, $\tpw(G)\le c\cdot \ecrw(G)$.

\subsection*{Acknowledgments}
\noindent
We thank the anonymous reviewers for comments to improve the original manuscript.

\noindent An extended abstract of this paper appeared in the proceedings of WG 2023~\cite{ChangKL2023}. Y. Chang, O. Kwon, and M. Lee are supported by the National Research Foundation of Korea (NRF) grant funded by the Ministry of Science and ICT (No. NRF-2021K2A9A2A11101617 and RS-2023-00211670). O. Kwon is also supported by Institute for Basic Science (IBS-R029-C1).

\subsection*{Declarations}
\noindent
Declarations of interest: none.

\end{document}